\documentclass[11pt]{article}

\usepackage[margin=1in]{geometry}
\usepackage{setspace}
\usepackage{amsmath,amssymb}
\usepackage{amsthm}
\usepackage{graphicx}
\usepackage{xcolor}
\usepackage{booktabs}
\usepackage{float}
\usepackage{caption}
\usepackage{hyperref}
\hypersetup{
  colorlinks=true,
  linkcolor=black,
  citecolor=black,
  urlcolor=black,
  pdftitle={A Gate-and-Menu Theory of Collective Tourism Brand Value},
  pdfauthor={Johan Fourie},
  pdfkeywords={collective reputation; tourism brand; destination governance; weaker-link public good}
}

\usepackage[authoryear,round]{natbib}

\theoremstyle{plain}
\newtheorem{proposition}{Proposition}

\newtheorem{lemma}{Lemma}
\theoremstyle{definition}

\title{A Gate-and-Menu Theory of Collective Tourism Brand
Value\thanks{This paper began with a reader's question on my blog, ourlongwalk.com,
for which I am grateful. Statement on the use of artificial intelligence: large
language models were used throughout this work, in two roles. As measurement
instruments, Anthropic's Claude (Opus) and OpenAI's Codex (gpt-5.5) code the
paper's constructs from public text under a fixed rubric
(Sections~\ref{sec:measurement}--\ref{sec:results}); every prompt, setting and
output is retained in the replication package. As research assistance -- through
Claude Code and the Codex command-line interface, alongside smaller tools such as
refine.ink -- they drafted and edited prose under my direction, wrote and checked
the analysis code, adversarially verified the proofs and checked bibliographic
records. All modeling choices, interpretations and conclusions are mine, as are
all errors.}}
\author{Johan Fourie\thanks{Department of Economics, Stellenbosch
University. Email: johanf@sun.ac.za.}}
\date{}

\begin{document}
\maketitle

\begin{abstract}
\noindent A collective tourism brand is jointly produced by assets managed by
different custodians, so destination managers must decide which assets require
collective protection. We model a destination as a \emph{gate} times a
\emph{menu}, $V=G\cdot A$. The gate combines essential inputs and
non-reproducible anchors as a weaker-link public good; the menu combines many
reproducible attractions through love of variety. Deterioration of a binding gate
element has a non-vanishing, brand-wide effect, whereas an individual menu
attraction's share vanishes as the menu grows. Aumann--Shapley accounting allocates
the jointly produced brand value across assets. Under stated surplus and non-use-value
conditions, a decentralized custodian weakly under-maintains physical gate quality
holding other qualities fixed. We measure the framework's constructs from guide text
with two language models across 166 economies, subject the gate labels to a partial
external check against the World Heritage List, and illustrate a within-brand
contamination ordering with four French wine regions.
\end{abstract}
\vspace{1em}\noindent\textbf{Keywords:} collective reputation; tourism brand;
destination management; destination governance; weaker-link public good; text as data

\vspace{0.25em}

\vspace{0.5em}

\newpage

\section{Introduction}\label{sec:intro}

Destination managers promote a shared name but rarely control all the assets that
produce its value. A national tourism brand is the primary example. A country's
regions trade under one name. A visitor first decides whether to come to the
country and only then how to divide a trip among its regions, so the value of
every member's asset depends on inputs -- safety, infrastructure, prominent
anchors -- managed by separate public and private custodians.

Collective brands in other markets show the consequences. A shared name can
raise every member's return where members are relatively homogeneous, and protect
no one where a heterogeneous set of members dilutes the signal
\citep{schatt2025collective, winfreemccluskey2005}. Existing tourism research does
not yet offer an integrated formal account of when a destination umbrella protects
its members, which assets destination managers should coordinate around, and what
an individual attraction is worth inside the umbrella.

This paper asks two questions. First, when does a shared destination brand
protect its members, and when does it cease to? Second, what is an individual
attraction worth within the brand? The questions matter because tourism is, for
many countries, a major export sold under a shared reputation, and because many of
the assets that constitute it -- a prominent mountain, a major park, a functioning
gateway city -- are controlled by decentralized custodians that do not capture
their cross-regional effects. The two questions are linked: a theory of when the
shared brand protects members also identifies which assets require collective
management.

One approach is to treat the brand as a complementary bundle and to ask
how complementary it is. This approach is incomplete. If every attraction enters
one aggregator symmetrically, then a
minor museum and an iconic landmark are scaled by the same logic, and a model with
enough complementarity will let a small attraction matter as much as a large one.
That is implausible. The distinction is that a tourism brand is built from
two kinds of asset that behave differently. The first kind are
\textit{anchors} and \textit{essential inputs}: a small number of trip-generating
icons -- a celebrated monument, a major national park, a famous coastline --
together with the conditions that make a trip feasible at all (safety, electricity,
functioning airports and borders). These are scarce, hard or impossible to
reproduce, and they enter the weaker-link side of the structure set out below. If safety deteriorates sharply or an anchor is ruined, the value of the
whole destination falls, for every region, not just the one affected.
The second kind are \textit{modular attractions}: a large set of reproducible
goods -- museums, minor parks, wine estates, festivals -- each of which occupies
visitor time after arrival. These enter a love-of-variety menu. Add one and the
destination is marginally more valuable; in a large regular menu, removing or
degrading one changes the destination only through its small menu share.

We formalize this as a destination value function that is the product of a
\textit{gate} and a \textit{menu},
\[
  V = G(g) \cdot A(a),
\]
where the gate $G$ combines the essential inputs and anchors $g$ through a
\textit{weaker-link} aggregator, and the menu $A$ combines the modular attractions
$a$ through a love-of-variety CES. In a weaker-link aggregator, lower-quality
elements receive greater weight. Degrading a gate element therefore raises its own
value share. The
weaker-link form \citep{hirshleifer1983weakest, cornes1993dyke, corneshartley2007}
spans the cases used in this debate: at one extreme quality
simply multiplies (the Cobb-Douglas, or ``O-ring'', case of
\citet{kremer1993oring}); at the other the brand
is determined by its lowest-quality anchor (the Leontief case). The contribution of the gate-and-menu form is that the
asymmetry between the two kinds of asset is now structural. A shock to a gate
element affects the marginal value of all members and remains first-order for a
binding gate element; a change to an individual menu attraction reaches participation through a
share that vanishes as a regular menu grows. No tuning of a
single elasticity is required, and a minor museum does not dominate a large regular
menu, because it is in the menu and not in the gate. Thus the
O-ring logic applies to the gate, not to the whole
brand.

The model also treats the gate as a reputation.
The reason a visitor weighs the whole bundle when deciding whether to come is that
the gate is largely a shared belief about the country: whether it is safe, whether
it functions and whether its icons are intact. We represent that belief with a
simple reputation formula following \citet{tirole1996collective}: it turns a
weaker-link summary of how the country's anchors are run into an umbrella
reputation. It is not a full model of individual reputational histories, but it
keeps the reputation and the physical gate governed
by the same weakest-member logic, and it builds in two features of collective
reputations that the quality-aggregation model lacks. First, the umbrella is
concave by construction: when the brand's reputation is high, a marginal
improvement in aggregate conduct $Q$ has a small effect, but an anchor allowed to
deteriorate, or a gateway
city that becomes unsafe, damages the brand for everyone. Second, the gate is a
\textit{commons}. Each custodian captures only
its own region's tourism income and bears the full cost of maintaining physical
quality, so under the conditions derived below the shared physical-gate benefit is
weakly under-supplied, holding other qualities fixed. In the Leontief limit, the brand
is constrained by the lowest-quality gate element.

This structure makes the shared brand's asymmetry precise: a homogeneous, well-maintained gate
protects every member, while a member allowed to deteriorate damages the concave
umbrella disproportionately. The closest published pattern of this kind comes
from collective wine brands, where the umbrella of a homogeneous region raises
every member's return and the premium of the most heterogeneous region is
negative \citep{schatt2025collective}. Section~\ref{sec:results} uses those
estimates as a within-country illustration of the paper's measurement, not the
main object of analysis. The imposed kernel also reproduces a concavity documented in that
literature -- reputation premiums shift from collective to individual as quality
rises \citep{costanigro2010nested}.

The analysis makes four theoretical contributions and one measurement contribution;
the first two and the fourth answer the protection question, the third answers
the valuation question.
(i) It splits a tourism brand into a weaker-link gate of anchors and essential
inputs and a love-of-variety menu of modular attractions, and it grounds that split in a
basic property of demand: gate assets act directly on whether a visitor comes to the
country, while menu assets primarily affect how a visitor allocates time once there
and reach participation through the menu's inclusive-value share
(Sections~\ref{sec:model}--\ref{sec:asymmetry}). The result is a structural
asymmetry in how shocks propagate. (ii) It represents the gate with a reputation formula
built on the same weakest-member logic; reputation and physical quality can co-vary
when conduct and quality rank together, but the governance theorem conditions on
conduct (Section~\ref{sec:reputation}). (iii) It values assets with Aumann--Shapley
accounting -- a standard method for splitting the value of a jointly produced good
among the inputs that produce it -- and shows that a binding gate anchor earns a
non-vanishing share of brand value while any one menu attraction earns only a
small share, one that shrinks as the menu grows (Section~\ref{sec:valuation}).
(iv) It identifies conditions under which a gate custodian, holding other qualities
fixed, weakly under-maintains its asset, measures the resulting welfare gap and
frames the policy response as a choice
among imperfect coordination arrangements rather than an appeal to a benevolent
planner (Section~\ref{sec:commons}). (v) Following \citet{ludwig2025llm}, it uses
language models as measurement instruments to sort attractions into gate and menu and
to score how homogeneous a country's brand is from public guide text, reporting
cross-model agreement for 166 economies selected by a fixed sample rule and
subjecting the gate labels to a partial external check against the UNESCO World
Heritage List; because the contamination ordering is within-brand, it is illustrated
on French wine regions
(Sections~\ref{sec:measurement}--\ref{sec:results}). Together these results give
destination managers a sequence for brand stewardship: classify assets by their
role in participation, diagnose the binding gate elements and coordinate
maintenance where benefits cross custodial boundaries. They also give researchers
a way to identify which assets determine a shared tourism brand and what each
asset is worth within it. Proofs are in the appendix.

\section{Related literature}\label{sec:lit}

The paper relates to several literatures.

\paragraph{Aggregation technologies for public goods.} The idea that the social
output of many contributions need not be their sum is due to
\citet{hirshleifer1983weakest}, who distinguished summation, best-shot (the
maximum) and weakest-link (the minimum) technologies. \citet{cornes1993dyke} and
\citet{corneshartley2007} generalized this to a continuum of CES weaker-link
technologies between summation and the Leontief limit, in which below-average
contributions disproportionately depress the total without the strict lexicographic
rule of pure weakest-link. The gate uses the complementary sub-range of this family
-- from the Cobb-Douglas (O-ring) case down to the Leontief limit, that is
$\eta < 0$, not the summation end. It lets lower-quality anchors matter more without
letting any matter infinitely. And it applies only to the essential subset of
assets, so the large set of menu attractions never governs the brand.

\paragraph{Collective reputation.} \citet{tirole1996collective} models group
reputation as an aggregate of members' individual reputations under noisy
monitoring, generating multiple steady states, history dependence and an
asymmetry between losing reputation and rebuilding it. The present paper uses that
logic to discipline a reduced-form reputational kernel rather than solving the
full dynamic posterior problem. The static quality-reputation tradition
\citep{shapiro1983premiums} and the collective and common-property treatments
\citep{winfreemccluskey2005, nockestrausz2023} supply the free-riding logic: where
members cannot be distinguished under a shared brand, each chooses quality below the
group optimum and the reputation is depleted. In industrial organization the same
shared name is an \textit{umbrella brand}. It can serve as a signal backed by a
bond \citep{wernerfelt1988umbrella}, use a good's reputation in one market to
signal quality in another \citep{choi1998brand}, mitigate moral hazard and speed
the building of a reputation when its members are aligned
\citep{cabral2009umbrella}. Yet a single low-quality member can dilute it for the
others \citep{bacchiega2024dilution}. In a single-owner brand that
dilution may even be a deliberate choice to monetize the name; under a collective
brand no single owner makes or prices that trade-off. The quality spillover we
study is therefore an uninternalized externality rather than an optimization.
Empirically, collective reputation premiums are concave and diminish toward the
individual as quality rises \citep{costanigro2010nested, castriotadelmastro2015,
gergaud2017netbenefits}, and the closest tourism-relevant evidence
\citep{schatt2025collective} shows the umbrella protecting homogeneous members and
losing value under heterogeneous ones. \citet{nie2024collective} models a single
low-quality operator degrading a Chinese destination's shared rating. The same
mechanism underlies the trade literature on national reputation as an export asset,
where a country's reputation for quality affects the demand facing each of its
exporters and one firm's conduct affects the national label
\citep{chisik2003reputational}. We add to this
a complementary-bundle structure -- the gate and menu -- and an explicit
cross-regional externality, and we model the common-pool externality implied by
these mechanisms.

\paragraph{Anchors, keystones and the destination as a bundle.} A few
salient sites affect how visitors perceive and move through a place. In retail
and agglomeration this is the anchor-tenant logic, where a large anchor store
generates the visitor flow that
co-located, even substitute, tenants share, and the landlord internalizes that
externality by subsidizing the anchor's rent
\citep{konishisandfort2003, gouldpashigian2005}. The distinction between a
destination's core attractions and its supporting attributes is itself long
established in tourism economics \citep{crouchritchie1999}; the gate-and-menu form
maps it into the multiplicative-versus-additive split that generates the
asymmetry. \citet{candelafigini2010, candelafigini2012} develop
the tourism destination as an economic district supplying a bundle of complementary
local goods. The menu side preserves the standard Dixit--Stiglitz gain from
variety \citep{dixitstiglitz1977}: visitors value the range of goods at a
destination. Their Coordination theorem shows that, absent coordination among the
independent firms supplying the complementary bundle, the bundle is mispriced and
underprovided, a complementary-monopoly problem. Our gate-and-menu form gives this
bundle an internal structure and an asymmetry, and the coordination failure
reappears as under-maintenance of the gate.

\paragraph{Superstars and non-reproducible assets.} That a small difference in
quality can map into a large difference in returns when output is jointly consumed
is \citeauthor{rosen1981superstars}'s \citeyearpar{rosen1981superstars} economics
of superstars; the scarcity rents of non-reproducible assets are Ricardian, and
the quality-discrimination logic is \citet{mussarosen1978}. This literature provides
a formal basis for treating unique natural and cultural assets as scarce:
anchors are in fixed supply and earn scarcity rents, which is why they belong in the
multiplicative
gate, while a constructed museum enters a competitive, reproducible category and
belongs in the menu.

\paragraph{Place and nation branding.} A parallel literature in marketing treats the
country itself as a brand -- a managed reputational asset that shapes how outsiders
evaluate everything sold under its name \citep{kotlergertner2002, anholt2007brand}.
A growing body of work now reviews this field of place branding
\citep{swain2024place}. That tradition emphasizes identity and positioning but is
largely informal about mechanism. The gate-and-menu model gives the managed
reputation an economic structure: it identifies which assets constitute the brand (the gate),
why a lapse at one of them affects all (multiplicative complementarity) and how
brand-value shares can be accounted for within the gate and menu blocks. The same
country name can also act as a collective signal in product markets; here the
signal is placed at the tourism participation margin rather than at the product
choice margin. For destination management, this converts place branding from only
a communication and positioning problem into an asset-governance problem: the
organization promoting the name must identify and coordinate the custodians whose
assets determine whether visitors enter the destination at all.

Finally, the paper relates to the quantitative spatial tradition on tourism and
development. \citet{fabergaubert2019} model tourism in a spatial equilibrium where
destinations are substitutes and cross-region spillovers largely cancel through
factor mobility; \citet{magontier2024coastal} show decentralized coastal
governments over-develop because they ignore the amenity their building destroys
for neighbors. And \citet{balboni2025harm} studies durable investment, lock-in
and policy myopia. More recent work lets neighboring destinations be complements as
well as substitutes, with spatial spillovers that need not cancel
\citep{jiao2025spillover}, though it gauges competitiveness across
destinations rather than the within-country propagation through a shared gate
that we study. We adopt the spatial-demand structure from
\citet{fabergaubert2019}, the decentralized-externality logic from
\citet{magontier2024coastal} and the durable-investment setting in
\citet{balboni2025harm}. But the mechanism here is a demand-side collective
reputation that need not cancel within the country.

\paragraph{Environmental and recreation valuation.} Valuing natural and cultural
assets is the subject of a large environmental-economics literature organized around
total economic value and the methods that recover it: contingent valuation, choice
experiments, the travel-cost method and ecosystem-service accounting
\citep{krutilla1967conservation, freeman2014measurement, carsonhanemann2005cv}. Its
revealed-preference branch often uses weak complementarity
\citep{maler1974environmental}, while multi-site recreation models allow detailed
substitution and complementarity across sites \citep{phaneufsmith2005recreation}.
Our framework complements this work rather than replacing it: it isolates a shared
quality argument at the national participation margin and uses Aumann--Shapley
accounting to partition the jointly produced brand value
\citep{herriges2004whats}. Section~\ref{sec:envval} sets out how the two compose.

\paragraph{Language models as measurement.} On the empirical side the paper uses the
framework of \citet{ludwig2025llm}, who distinguish predicting outcomes from text
(where training leakage threatens validity) from the measurement of economic
concepts (where a validation sample delivers consistent downstream inference). We use
a language model only in the second, safer mode -- to measure the gate-menu structure
and brand homogeneity that the data do not record directly. The broader literature on
treating text and images as economic data \citep{gentzkow2019text, dell2025deep}
supplies the tools; the discipline we take from \citet{ludwig2025llm} is that
cross-model agreement is only reliability and that validity ultimately requires an
external gold-standard sample.

\section{A gate-and-menu model of destination value}\label{sec:model}

\subsection{Assets, gate and menu}

A country contains a set of tourism assets. The primitive distinction is not
``large'' versus ``small'', but the margin on which the asset enters demand. Each
asset has a participation-salience component and a reproducibility component. Assets
that are essential for taking the trip, or that are non-reproducible reasons for the
trip, affect the participation decision: a visitor may not come to the country if
they fail. Assets that are reproducible and modular affect mainly the conditional
allocation of time after the trip has been chosen and reach participation through
the menu's inclusive value. All qualities are dimensionless indices normalized by
asset-specific full-quality benchmarks. The first group forms the \textit{gate set}
$K$, with quality levels $g=(g_k)_{k\in K}$ and
$g_k\in[\underline g_k,1]$, where $0<\underline g_k<1$. The positive lower bound
makes the custodians' feasible quality sets compact; limits toward zero are used
only for accounting and boundary arguments. The second group forms the
\textit{menu set} $M$, with normalized quality levels $a=(a_i)_{i\in M}$ and
$a_i\in(0,1]$. The weights below are defined relative to these normalizations, so
comparisons such as ``lowest quality'' do not depend on the original units of an
asset's engineering or service measure.

This motivates a simple sorting rule. Let $b_n$ denote asset $n$'s participation
salience and let $\lambda_n$ denote its reproducibility, with lower $\lambda_n$
meaning harder to reproduce. Given outside options, suppose there is a threshold
$\bar b(\lambda_n)$ such that assets with $b_n \ge \bar b(\lambda_n)$ enter the
participation constraint directly and are gate assets; assets with
$b_n<\bar b(\lambda_n)$ enter primarily through conditional itinerary utility and
affect participation only through the menu aggregate. The threshold is
decreasing in non-reproducibility: a unique natural or cultural asset needs less
incremental salience to become a reason for the trip than a reproducible attraction.
This is a disciplined partition, not a full derivation from a complete demand
system. Section~\ref{sec:measurement} makes the distinction operational from guide
text.

\textbf{The gate} combines its $|K| \ge 2$ assets with a weaker-link CES
aggregator -- a constant-elasticity-of-substitution function, the standard flexible
way to combine inputs with an adjustable degree of complementarity:
\begin{equation}\label{eq:gate}
  G(g) = \left[ \sum_{k \in K} \phi_k\, g_k^{\eta} \right]^{1/\eta},
  \qquad \phi_k > 0, \ \sum_k \phi_k = 1, \ \eta < 0.
\end{equation}
With $\eta < 0$ the gate assets are gross complements: the marginal product of each
gate asset is increasing in the quality of the others, and degrading a given gate
asset raises its own value share. Cross-asset rankings additionally depend on the
weights $\phi_k$. Two boundary cases are useful. As
$\eta \to 0^{-}$ the gate tends to Cobb-Douglas, $G \to \prod_k g_k^{\phi_k}$ (the
O-ring case, in which quality multiplies); as $\eta \to -\infty$ it tends to Leontief,
$G \to \min_k g_k$ (pure weakest-link). The interior $\eta < 0$ is the weaker-link
case \citep{corneshartley2007}: under equal weights the lowest gate quality matters
most, but does not by itself determine the whole gate as it would in the pure
weakest-link limit. With unequal weights, the relevant ranking is weighted quality.
Define the \textit{gate share}
\begin{equation}\label{eq:gateshare}
  \gamma_k(g) \equiv \frac{\phi_k g_k^{\eta}}{\sum_{l} \phi_l g_l^{\eta}},
  \qquad \sum_k \gamma_k = 1, \qquad
  \frac{\partial \ln G}{\partial \ln g_k} = \gamma_k.
\end{equation}
Under complementarity ($\eta < 0$) the gate share of an asset rises as it
degrades, and more the weaker it gets; shares rank assets by weighted quality
$\phi_k g_k^{\eta}$, so with equal weights the lowest-quality gate asset carries
the largest share. The case $|K|=1$ is degenerate but useful: the gate is then just the quality
of the single essential asset, and the weaker-link results become equalities.

\textbf{The menu} is a love-of-variety CES aggregator,
\begin{equation}\label{eq:menu}
  A(a) = \left[ \sum_{i \in M} w_i\, a_i^{\rho} \right]^{1/\rho},
  \qquad w_i > 0, \ \rho = \frac{\nu - 1}{\nu} \in (0,1), \ \nu > 1,
\end{equation}
so the menu attractions are substitutes, $A$ is concave in each $a_i$ and there is
love of variety (holding quality fixed, a wider menu raises $A$). The additive case
$A = \sum_i w_i a_i$ is the limit $\rho \to 1$ and carries the same
qualitative results. Define the \textit{menu share}
$s_i(a) \equiv w_i a_i^{\rho} / \sum_l w_l a_l^{\rho}$, with
$\partial \ln A / \partial \ln a_i = s_i$ and $\sum_i s_i = 1$. Because the menu
contains many small items, the typical menu share is small: it is of order $1/|M|$
when the menu terms $w_i a_i^{\rho}$ are bounded above and away from zero,
uniformly in the size of the menu. We call this \textit{menu regularity} and
maintain it wherever a claim concerns the growing menu
(Propositions~\ref{prop:asymmetry} and~\ref{prop:shapley}); it rules out a dominant
menu term in those asymptotic claims. The items are substitutes in the
standard CES sense, but they are valuable as varieties because a larger set raises
$A$; for $\rho \in (0,1)$ their cross-partials in the aggregate are positive.

\textbf{Destination value} is the product
\begin{equation}\label{eq:value}
  V(g, a) = G(g) \cdot A(a).
\end{equation}
The multiplicative form is the formal content of ``anchors and the menu are
complements'': a strong gate raises the value of every menu attraction, and a
gate that tends to zero ($G \to 0$) makes destination value tend to zero for any
finite menu. The
host-country inputs (safety, infrastructure, stability) belong in $G$, so
that growth and state capacity raise the productivity of every asset at once.
More general aggregators $\Phi(G,A)$ would preserve the gate--menu cross-effect
whenever $\Phi_G>0$, $\Phi_A>0$, $\Phi_{GA}>0$, and $\Phi(0,A)=0$ in the interior;
gate--gate complementarity would additionally restrict $\Phi_{GG}$. The product form is the
tractable benchmark that makes the accounting result in Section~\ref{sec:valuation}
transparent.

\subsection{Demand}

There is a positive market-size scale $\bar M$, and the rest of the world is held
fixed. Arrivals to the country are
\begin{equation}\label{eq:arrivals}
  T(g, a) = \bar M \cdot D\!\left( V(g,a) \right), \qquad D' > 0,
\end{equation}
with participation elasticity $\psi(V) \equiv \partial \ln D / \partial \ln V > 0$.
Here $D$ is a demand intensity, not a probability, so the constant-elasticity
benchmark used for the aggregate results,
$D(V)=(V/\bar V)^{\psi}$, need not be capped at one. If $\bar M$ is instead read as
a literal finite population of potential visitors, this benchmark is local to the
range on which $D(V)\le1$. Because $V=GA$, menu quality affects participation too:
$\partial\ln T/\partial\ln a_i=\psi s_i$. The distinction between gate and menu is
therefore a direct, non-vanishing participation effect versus a diluted
inclusive-value effect, not a zero participation effect for the menu. Conditional
on arriving, visitors allocate the trip
across regions; regions remain substitutes on this intensive margin, as in
\citet{fabergaubert2019}, and nothing below depends on the intensive allocation
except where stated. A full gravity closure would add competing countries and
diversion to them; the claims below are therefore within-country and
partial-equilibrium claims.

\section{The structural asymmetry}\label{sec:asymmetry}

The central result is that the two asset classes propagate shocks differently, and
that the difference is built into the structure rather than tuned through a
parameter.

\begin{proposition}[Gate and menu shocks differ in magnitude and reach]\label{prop:asymmetry}
The elasticity of destination value to a gate asset is its gate share, and to a
menu asset is its menu share:
\begin{equation}\label{eq:elasticities}
  \frac{\partial \ln V}{\partial \ln g_k} = \gamma_k,
  \qquad
  \frac{\partial \ln V}{\partial \ln a_i} = s_i .
\end{equation}
Three asymmetries follow. (i) \textbf{Magnitude}: a gate asset's share $\gamma_k$
does not shrink as the menu grows -- the weakest gate asset's share satisfies
$\gamma_k \ge \phi_k$, and tends to one for a unique binding asset as $\eta \to
-\infty$ -- while under menu regularity a menu asset's share $s_i$ is of order
$1/|M|$. For a fixed gate and a growing menu, the same infinitesimal proportional
degradation therefore reduces value more through the weakest gate asset than through
a menu attraction by the local elasticity ratio $\gamma_k/s_i$, which grows without
bound. Finite changes are nonlinear but have the same asymptotic ordering for a fixed
proportional degradation. (ii) \textbf{Reach}: a gate shock
moves the marginal value of every other asset (both cross-partials are positive,
$\partial^2 V / \partial g_k\,\partial a_i > 0$ for all $i \in M$, and
$\partial^2 V / \partial g_k\,\partial g_l > 0$ for $l \ne k$) and it does so
strongly: a change in $g_k$ scales the
marginal value of every menu attraction by the elasticity $\gamma_k$, and of every
other gate asset at least as strongly. Cross-effects run in both directions, since
$V$ is a product; the asymmetry is quantitative, not one of sign: a menu shock
scales the marginal value of every gate asset only by the vanishing elasticity
$s_i$, and it leaves the gate aggregate $G$, the gate shares, and every value
elasticity to the gate unchanged. (iii)
\textbf{Extensive margin}: under menu regularity, adding a new attraction whose
menu term $w a^{\rho}$ is bounded raises destination value by a vanishing
fraction, $\Delta V / V \to 0$ as the menu grows, so a marginal attraction is a
small and eventually vanishing share of brand value.
\end{proposition}

Proposition~\ref{prop:asymmetry} resolves the identification problem that motivates the
model. In a single-aggregator bundle, a small attraction can be made to matter as
much as an anchor by choosing enough complementarity. Here the asymptotic difference
comes from block membership and menu regularity: a bounded menu term is divided by a
growing sum, whereas the weakest gate asset's share is bounded away from zero
independently of the menu and its gate multiplies the entire menu. Concavity alone
does not make an attraction's share small. The asymmetry does not
require the assets to differ in their own elasticities; it requires only that they
be sorted correctly into gate and menu, which is an empirical question
(Section~\ref{sec:measurement}).

The weaker-link form gives the gate an internal asymmetry.

\begin{proposition}[The binding anchor, softened]\label{prop:weakerlink}
For $\eta < 0$ and $|K| \ge 2$, the gate share $\gamma_k$ is strictly decreasing in
$g_k$: degrading an asset raises its own share, and more the weaker it gets. Shares
rank assets by weighted quality, $\gamma_k/\gamma_l = (\phi_k/\phi_l)(g_k/g_l)^{\eta}$,
so with equal weights the lowest-quality gate asset carries the largest share. As
$\eta \to -\infty$ the gate share concentrates on the binding set $I = \{l : g_l =
\min_j g_j\}$: $\gamma_k \to \phi_k/\sum_{l \in I}\phi_l$ for $k \in I$ (equal to
$1$ when the minimizer is unique) and $\gamma_k \to 0$ otherwise, while $G \to
\min_l g_l$ in every case. For finite $\eta < 0$ the weight on the worst-kept
asset rises as its quality falls without completely determining the gate on its
own, and because the gate set $K$ excludes the menu, no menu asset ever governs
the brand.
\end{proposition}

\section{The gate as a collective reputation}\label{sec:reputation}

So far the gate is a quality aggregator. Its economic content is also reputational:
prospective visitors form a belief about whether the country is safe, whether it
works and whether its anchors are intact. The reputation should load on the same
weaker-link object as the technological gate. Otherwise the paper would have two
different gates: a weaker-link CES in qualities and a linear average in conduct.

Let each gate asset be run at conduct $q_k \in (0,1]$, interpreted as the
probability that the asset or the operators behind it behave well (the statistic
below extends to zero conduct by continuity). Define the
weaker-link conduct statistic
\begin{equation}\label{eq:conduct}
  Q(q) = \left[ \sum_{k \in K} \omega_k q_k^{\eta} \right]^{1/\eta},
  \qquad \omega_k > 0,\ \sum_k \omega_k = 1,\ \eta < 0 .
\end{equation}
The benchmark sets $\omega_k=\phi_k$, but the notation separates conduct salience
from engineering quality. Prospective visitors do not observe conduct directly.
They observe public lapses with probability $x \in (0,1)$. Conditional on the
no-lapse signal, a simple reduced-form updating rule gives the umbrella
reputation as
\begin{equation}\label{eq:reputation}
  R(Q) = \frac{Q}{Q + (1 - Q)(1 - x)} .
\end{equation}
This expression is not a full Tirole posterior derived from individual histories
and a record law. It is a parsimonious reputational transformation of the
weaker-link statistic $Q$, not a separate linear average. Its concavity
(Lemma~\ref{lem:concave}) is a property of this one-parameter form, imposed to
reproduce the documented collective-reputation asymmetry
\citep{tirole1996collective}, not derived from informational primitives. It is not
an unconditional Bayesian law: an observed lapse would generate a different
posterior. In the reputation-adjusted specification, effective destination value is
\begin{equation}\label{eq:valueR}
  V_R(q,g,a)=R(Q(q))G(g)A(a)=R(Q(q))V(g,a),
\end{equation}
so the same weaker-link logic governs both margins: the technological gate is
bound by its weakest-quality member and the reputation by its weakest-conduct
member. These are the same asset when conduct and quality rank together -- the
principal case, in which a poorly maintained anchor is both low quality and poorly
operated -- but the model does not require it. The remaining propositions about
physical asset quality and menu quality condition on a fixed conduct profile $q$.
Thus $R(Q(q))>0$ is a multiplicative constant, arrivals are
$T_R(q,g,a)=\bar M D(V_R(q,g,a))$, and the previously derived $g$- and
$a$-elasticities are unchanged. The Aumann--Shapley calculation below likewise
allocates $V_R$ across physical asset qualities conditional on $q$.

If conduct is instead a choice variable, its separate elasticity is
\begin{equation}\label{eq:conductelasticity}
  \frac{\partial\ln V_R}{\partial\ln q_k}
  =\frac{Q R'(Q)}{R(Q)}\chi_k
  =\frac{1-x}{(1-x)+xQ}\chi_k.
\end{equation}
If the restriction $q_k=g_k$ were imposed, this term would be added to $\gamma_k$.
That endogenous-conduct governance model is not imposed here: the commons result
below concerns normalized physical quality $g_k$, holding conduct fixed.

\begin{lemma}[The concave reputation kernel]\label{lem:concave}
$R$ is strictly increasing and strictly concave in $Q$ on $[0,1]$, with $R(0) = 0$,
$R(1) = 1$, $R'(Q) = (1-x) / [\,(1-x) + xQ\,]^{2} > 0$ and $R'' < 0$. The
marginal effect of the aggregate conduct statistic $Q$ on reputation is small when
the brand's reputation is high, $R'(1) = 1 - x$, and large when it is not,
$R'(0) = 1/(1-x)$.
\end{lemma}

Lemma~\ref{lem:concave} states the asymmetry between improving and degrading the
aggregate conduct statistic, in static form. When the umbrella is already
high-quality ($Q$ near one), a marginal common-conduct improvement has a small
effect: the marginal improvement in $Q$ is worth $1-x$,
which is small when monitoring is good. But a member that lowers $Q$ does
disproportionate damage. The slope steepens as $Q$ falls. The same
concavity has been documented for collective wine brands, where the premium
shifts from collective to individual reputation as quality rises
\citep{schatt2025collective, costanigro2010nested}.

Because $Q$ is weaker-link, the conduct share of asset $k$,
\[
  \chi_k(q) = \frac{\omega_k q_k^{\eta}}{\sum_l \omega_l q_l^{\eta}},
\]
rises as $q_k$ falls. For an individual member,
$\partial R/\partial q_k=R'(Q)Q\chi_k/q_k$: the marginal effect of lowering that
same member's conduct rises as it deteriorates. Cross-member rankings also depend on
the weights $\omega_k$, so the lowest raw conduct need not have the largest marginal
effect when weights differ. This comparative static comes from combining collective
reputation with a weaker-link gate.

For two distinct gate members, the interaction is
\[
  \frac{\partial^2 R(Q(q))}{\partial q_k\,\partial q_l}
  = R'(Q)Q_{kl} + R''(Q)Q_kQ_l, \qquad k\ne l,
\]
where $Q_k=\partial Q/\partial q_k$ and
$Q_{kl}=(1-\eta)Q\chi_k\chi_l/(q_kq_l)>0$. The first term is weaker-link
complementarity: improving one member raises the return to improving another. The
second term is the concavity of the reputational kernel. The net cross-partial is
therefore not imposed by a single CES share; its sign and magnitude depend on the
relative strength of these two terms.

The dynamic collapse mechanism is not used as a theorem in this paper. If arrivals
respond nonlinearly to reputation near a trust threshold, and if maintenance is
partly revenue-financed, a shock to one gate anchor can lower $Q$, then $R(Q)$, then
arrivals and maintenance revenue for the whole gate. With memory in public records,
the return path is slower than the fall, as in Tirole's persistence result. The
static propositions below do not require multiplicity, so the paper does not assume
a collapse region.

\section{Aggregation and valuation}\label{sec:valuation}

\subsection{Within-country non-neutrality}

\begin{proposition}[Within-country non-neutrality and gate supermodularity]\label{prop:aggregate}
Fix the conduct profile $q$. Total arrivals fall when any physical gate asset is
degraded: $\partial \ln T_R / \partial \ln g_k = \psi\, \gamma_k > 0$. The gate
elements are complements in reputation-adjusted value,
$\partial^2 V_R / \partial g_k\, \partial g_l > 0$ for $k \ne l$, and gate and menu
are complements, $\partial^2 V_R / \partial g_k\, \partial a_i>0$; with the
constant-elasticity participation $D(V_R)=(V_R/\bar V)^{\psi}$, national tourism
value $V_R^{\psi}$ is supermodular in the gate. Fix any floor quality vector
$g^{f}$ with $\underline g_k\le g^{f}_k\le g_k$ for all $k$. Relative to the
floor baseline $g^{f}$ and holding conduct and the menu fixed, the normalized value of any
combined set of gate assets is therefore at least the sum of the normalized values
of those gate assets taken separately.
\end{proposition}

This is the contrast with the quantitative spatial tradition. In
\citet{fabergaubert2019}, a tourism boom in one region pulls mobile factors from
others, so local gains are offset elsewhere and the cross-region spillover nearly
cancels in the aggregate. The spillover here is a demand-side shift in how many
people come to the country. In the partial-equilibrium country model it does not
cancel within the country: degrading a gate anchor removes visitors from the
country, not only from the affected region. In a world model the global sign is
indeterminate because some visitors would be diverted to other countries. The
within-country externality remains the object of interest here.

\subsection{What an asset is worth}

This subsection determines what a single asset contributes to the brand as a whole,
using a standard tool from cooperative game theory. The economic asymmetry it
prices is a within-block one: binding
gate anchors receive non-vanishing brand-value shares, while individual menu
attractions receive shares of order $1/|M|$. For continuous quality, the appropriate accounting object is not a discrete
Shapley game that depends on an arbitrary baseline, but the Aumann--Shapley value
\citep{aumannshapley1974}. This is a static accounting exercise. Durable stocks,
discounting and recovery dynamics are outside this object. The method prices each
quality dimension by integrating its marginal contribution along a common expansion
path from the no-tourism origin to the observed brand, holding conduct fixed:
\[
  \alpha_n^{AS}(V_R\mid q) =
  \int_0^1 y_n
  \left.\frac{\partial V_R(q,z)}{\partial z_n}\right|_{z=ty}\,dt ,
\]
where $y=(g,a)$ and the expression is understood as the positive-limit path when a
gate coordinate approaches zero. The path is an accounting extension below the
positive feasible floor, not a feasible policy path. This allocation does not depend
on a discretionary restoration floor and exactly exhausts $V_R$ conditional on $q$.

\begin{proposition}[Aumann--Shapley accounting under a product index]\label{prop:shapley}
Under the reputation-adjusted structure \eqref{eq:valueR}, holding $q$ fixed, the
Aumann--Shapley allocation of brand value across physical gate and menu qualities is
\begin{equation}\label{eq:as}
  \alpha_k^{AS} = \frac{1}{2}\gamma_k V_R
  \quad (k \in K), \qquad
  \alpha_i^{AS} = \frac{1}{2}s_i V_R
  \quad (i \in M).
\end{equation}
The gate block and menu block each receive one half of brand value along the common
diagonal path, and assets split each block according to their gate and menu shares.
Thus a binding gate anchor receives a non-vanishing share of brand value, while a
menu attraction receives a share of order $s_i/2$, which is $O(1/|M|)$ under menu
regularity.
\end{proposition}

The one-half gate/menu split is a property of the degree-two product index and the
common diagonal path, not an independent economic asymmetry.\footnote{Along
$z(t)=(tg,ta)$, $G(tg)=tG(g)$ and $A(ta)=tA(a)$, so
$V_R(q,z(t))=t^2V_R(q,g,a)$. The
one-half block allocation is the integral of $t$ over the unit interval. A different
homogeneous structure or a different aggregator $\Phi(G,A)$ would change the block
split.} If conduct $q$ were itself included in the common expansion path, the
conditional half split would change. In that expanded accounting, define
$I_R(Q)=\int_0^1tR(tQ)\,dt$. The gate-quality block receives $GAI_R(Q)$, the menu
block receives the same amount, and conduct coordinate $q_k$ receives
$\chi_kGAQ\int_0^1t^2R'(tQ)\,dt$; integration by parts shows that these allocations
sum to $V_R$. We do not use this expanded object because conduct is held fixed in
the physical-quality governance model. The economic asymmetry is within blocks. An anchor's value-share is not its
ticket revenue or the spending of visitors who physically reach it. It is its
Aumann--Shapley share of the brand index (or of a money-metric welfare or revenue
object proportional to that index). A minor museum can still be valuable in absolute
terms when the destination is large, but its share of the brand is small.

The Aumann--Shapley share answers an accounting question -- how the standing level
of brand value decomposes across assets -- and not the marginal policy question.
What maintaining an asset is worth at the margin, including its effect on other
regions' arrivals, is the welfare wedge $\tau_k^{W}$ of
Section~\ref{sec:commons}, a distinct object: protection priorities follow the
wedge, while the shares price the standing stock. The two rank assets the same way
only under conditions the paper does not impose.

\subsection{Relation to environmental valuation}\label{sec:envval}

Environmental economics has a large set of methods for valuing natural and cultural assets.
Its organizing concept is total economic value -- use value alongside option, existence
and bequest value \citep{krutilla1967conservation, freeman2014measurement} -- recovered
by contingent valuation, choice experiments, the travel-cost method and
ecosystem-service accounting \citep{carsonhanemann2005cv, adamowicz1994combining,
costanza1997value}. Applied to nature-based tourism, these methods estimate what an
asset is worth to those who use or would preserve it
\citep{mukanjari2021valuation, naidoo2005biodiversity, choi2010heritage}. The
framework here values a different object.

The standard revealed-preference apparatus is organized by \emph{weak
complementarity} \citep{maler1974environmental}: an asset's quality matters through
the individual's use of the relevant good. Multi-site recreation models allow
detailed cross-site substitution and complementarity \citep{phaneufsmith2005recreation},
so the claim is not that such effects are absent from that literature; it is that the
gate acts at the participation margin. Degrading one anchor lowers the attractiveness
of visiting the country at all, including for travelers whose itinerary lay in other
regions. The two approaches must therefore be partitioned, not added. If a
stated-preference or travel-cost design already includes a national gate attribute --
safety, border functioning, the condition of a prominent anchor -- the brand channel is
partly inside its estimated total economic value; if it values own-site use with the
gate held fixed, the Aumann--Shapley allocation of Proposition~\ref{prop:shapley}
decomposes the jointly produced brand term into asset shares without double-counting.
The contribution is to identify the cross-asset brand term, show why it can be large
for binding gate anchors, and connect it to the governance failure. Under the
conditions stated below, a custodian that does not capture the full domestic marginal
surplus chooses weakly lower physical quality, holding the other qualities fixed.

\section{The gate is a commons}\label{sec:commons}

The assets are not chosen by a planner. Each is run by a custodian -- a parks
agency, a municipality, a provincial government, a concession holder, a developer
-- who captures the local benefits of quality and the local rents of reducing an
asset's quality, but who observes only the tourism income accruing to its own region.
Fix the conduct profile $q$. Custodian $k$ of a gate anchor chooses normalized
physical quality $g_k$ to maximize
\begin{equation}\label{eq:custodian}
  \pi_k(g_k;g_{-k},a,q)=B_k(1-g_k)+\mu_kN_k(q,g,a)-c_k(g_k),
\end{equation}
where $B_k$ is the development or extraction benefit from degrading below pristine,
$\mu_k$ the captured net return per arrival, $N_k$ region $k$'s arrivals, and
$c_k$ the cost of maintaining quality. The term $\mu_kN_k$ is captured private or
fiscal revenue, not social surplus. For welfare accounting let $S_j(N_j)$ denote
domestic surplus generated by tourism in region $j$, including producer surplus,
domestic visitor surplus and fiscal revenue net of tourism-production costs, with
$S'_j\ge0$. To avoid double-counting, $S_j$ excludes the asset-specific development
benefit $B_k$ and maintenance cost $c_k$. Let $\zeta_k(g_k)$ denote non-use, option,
identity or existence value omitted from the custodian's objective. The domestic
welfare objective is explicitly
\begin{equation}\label{eq:welfare}
  W(q,g,a)=\sum_{k\in K}\bigl[B_k(1-g_k)-c_k(g_k)\bigr]
  +\sum_jS_j(N_j(q,g,a))+\sum_{k\in K}\zeta_k(g_k).
\end{equation}
Thus $B_k$ and $c_k$ receive the same domestic valuation in the private and social
problems; alternative ownership, transfer or environmental-cost assumptions would
add corresponding terms to the wedge.

\begin{proposition}[Conditional under-maintenance of a gate asset]\label{prop:commons}
Fix $(q,g_{-k},a)$. Suppose $N_j=h_jT_R$, where $h_j\ge0$ is independent of $g_k$
and $h_k>0$. Throughout the feasible quality interval, suppose the
custodian's captured return does not exceed its own-region domestic marginal surplus,
$S'_k(N_k)\ge\mu_k$, and non-use value is nondecreasing,
$\partial\zeta_k/\partial g_k\ge0$. Assume the private problem \eqref{eq:custodian}
and the one-coordinate social problem \eqref{eq:welfare} have nonempty interior
maximizer sets. Then the social argmax dominates the private argmax in the strong set
order; in particular, their smallest and largest selections are weakly ordered. With
unique maximizers, $g_k^*\ge g_k^{dec}$. The
welfare wedge between the social and private marginal benefit of conserving gate
asset $k$, at the prevailing quality profile, is
\begin{equation}\label{eq:wedge}
  \tau_k^{W}
  = \sum_j S'_j(N_j)\frac{\partial N_j}{\partial g_k}
  - \mu_k\frac{\partial N_k}{\partial g_k}
  + \frac{\partial \zeta_k}{\partial g_k}.
\end{equation}
It is strictly positive over the comparison interval -- and unique interior optima
are strictly ordered -- if the weak inequalities above hold and at least one of the
following does: $S'_k>\mu_k$; another region has $S'_j>0$ and $h_j>0$; or
$\partial\zeta_k/\partial g_k>0$. This
$\tau_k^{W}$ is the first-best
wedge; a marginal cost of public funds or incidence on foreign visitors scales the
corrective instrument rather than leaving it equal to $\tau_k^{W}$
(Section~\ref{sec:policy}). With
non-distortionary public funds and no incidence correction, any transfer schedule
whose marginal slope at the target equals $\tau_k^W$ implements the domestic welfare
first-order condition there. A constant maintenance subsidy of $\tau_k^W$ per unit
of gate quality -- equivalently a tax of $\tau_k^W$ per unit of degradation -- is
one such schedule. A per-arrival revenue-sharing increment instead has units of
money per arrival and must satisfy
$\Delta\mu_k=\tau_k^W/(\partial N_k/\partial g_k)$ at the target. Under the weaker-link gate
(Proposition~\ref{prop:weakerlink}), in the Leontief limit the brand is bound by
whichever custodian chooses the lowest quality; among otherwise identical
custodians, the quality argmax correspondence at finite $\eta<0$ is nondecreasing in
the captured return $\mu_k$ and nonincreasing under upward shifts in marginal
development benefit or maintenance cost. For convergent selections this weak order
is preserved as $\eta\to-\infty$, when national tourism is bound by the lowest chosen
quality.
\end{proposition}

This is the common-pool asset implied by \citet{winfreemccluskey2005} and the wine evidence
\citep{schatt2025collective, castriotadelmastro2015}.
The physical gate is a common-pool asset: every region benefits from it, every
custodian can damage it and no custodian pays for all of the harm it does the
others. The fixed-conduct reputation factor scales this physical-quality wedge; an
endogenous conduct wedge would additionally use \eqref{eq:conductelasticity}. The
closest tourism analogue is the common-pool treatment of \citet{pintassilgosilva2007},
where open access to a shared environmental quality leads to its
overexploitation; the shopping-mall literature provides the single-owner comparison, a
landlord who internalizes the anchor externality by subsidizing it
\citep{gouldpashigian2005}, which is the single-owner benchmark for internalizing
the externality.
Two features of the present model refine the standard commons result. First, the
elasticity and the effect relative to a comparable gate asset vanish for an
individual menu attraction as a regular menu grows. Its absolute arrival or monetary
externality need not vanish: with comparable menu terms it is proportional to
$|M|^{\psi/\rho-1}$ and requires $\psi<\rho$ to converge to zero. Coordination
should therefore prioritize the gate on the normalized margin while still comparing
absolute benefits with costs. As $\eta\to-\infty$, the tourism-spillover component
of the correction concentrates on the unique binding gate asset; direct terms such
as $\partial\zeta_k/\partial g_k$ need not. For finite $\eta$ the other gate assets
still have positive but smaller tourism shares. Second, Proposition~\ref{prop:weakerlink}
ranks finite-$\eta$ shares by weighted quality and identifies the lowest normalized
quality as the binding constraint in the Leontief limit. Among
otherwise comparable custodians, the comparative statics in
Proposition~\ref{prop:commons} connect lower captured returns or higher marginal
development and maintenance costs to weakly lower selected quality. Coordination
therefore matters especially when custodians are heterogeneous, which
is \citeauthor{corneshartley2007}'s \citeyearpar{corneshartley2007} point for
weaker-link goods and \citeauthor{candelafigini2012}'s
\citeyearpar{candelafigini2012} Coordination theorem for the destination bundle.
Conditional on comparable salience and intervention costs, a proportional quality
improvement at the worst-kept anchor yields more national tourism value than the
same proportional improvement at one already secure. The return per dollar still
depends on the cost of improving each asset.

The model identifies the wedge; it does not imply that a central authority is always
the best institution. The first-best transfer result above assumes non-distortionary
public funds. With a marginal cost of public funds, a maintenance subsidy must pass
the usual second-best test: the domestic marginal benefit of the gate improvement
must exceed the fiscal cost inclusive of the financing distortion
\citep{sandmo1975externalities, atkinsonstern1974, ballardfullerton1992,
diamond1973externalities}. Incidence also matters. If foreign visitors bear part of
a tourism tax, a national objective may include tax exporting, while a global
welfare objective would count foreign visitor surplus and may imply a lower tax.
A self-governing club of custodians with monitoring and sanctions, in the sense of
\citet{ostrom1990governing}, can implement the same wedge when members can observe
gate quality and discipline the custodian who binds the weak link. Benefit-share
transfers can be interpreted as Lindahl-style prices for an impure public good; more
formal demand-revealing mechanisms are theoretically relevant but informationally
demanding \citep{clarke1971multipart, grovesledyard1977, cornes1996theory}. A
national authority is useful only under the same condition: it must have enough
information and independence to target the binding gate asset rather than the most
visible or politically connected one. Capture is endogenous to the wedge: the
custodian whose asset binds the gate may also be the actor most able to block
enforcement.

\section{From theory to measurement}\label{sec:predictions}

The theory yields four predictions. (P1) \textbf{Asymmetry}: a negative shock to a
gate anchor lowers national demand and the marginal value of all assets, while a
menu attraction is a vanishing share
of the brand (Propositions~\ref{prop:asymmetry}--\ref{prop:weakerlink}; the
arrivals statement is Proposition~\ref{prop:aggregate}). (P2)
\textbf{Conduct contamination}: a low-conduct member lowers $Q$ and weakens the
umbrella for all members; because the reputational kernel is concave, a given
reduction in $Q$ has greater marginal damage when $Q$ is already low
(Lemma~\ref{lem:concave}). With equal conduct weights and a fixed mean conduct level,
a mean-preserving spread also lowers the symmetric concave CES statistic $Q$.
Collective wine brands provide the closest published pattern of this kind
\citep{schatt2025collective}, and Section~\ref{sec:results} uses it only as an
illustration. Translating $R(Q)$ into a price premium requires the additional
empirical assumption that the premium is increasing in umbrella reputation; its
sign is not a theorem of the model. (P3)
\textbf{Valuation}: Aumann--Shapley accounting assigns non-vanishing within-gate
shares to binding gate anchors and small within-menu shares to individual menu
attractions (Proposition~\ref{prop:shapley}).
(P4) \textbf{Commons}: under the conditions of
Proposition~\ref{prop:commons}, physical gate quality is weakly under-maintained in
the conditional one-asset comparison; in the Leontief limit and among otherwise
comparable custodians, national tourism binds at the lowest selected quality.
Figure~\ref{fig:asymmetry} illustrates the two
halves of the asymmetry: the concave umbrella, and the aggregate effect of a gate
shock against the small-share effect of an individual menu shock.

The empirical contribution is to make these predictions operational. Each rests on
sorting attractions into gate and menu and on a measure of how homogeneous a
country's gate is, and both are latent. In principle the classification is identified
by reason-for-visit surveys (anchors are the primary reason for the trip),
visitor-flow and length-of-stay data (anchors generate inbound flow and anchor
itineraries), and network centrality among co-visited
attractions (anchors are central, modular attractions peripheral, a proxy for $\phi_k$
versus $w_i$). These data are not available at cross-country scale. We therefore measure the
constructs from text with language models (Section~\ref{sec:measurement}) and,
because prediction P2 is a within-brand claim, illustrate it on the four French
wine regions with the closest published collective-brand premiums
(Section~\ref{sec:results}).

\begin{figure}[htbp]
\centering
\includegraphics[width=\linewidth]{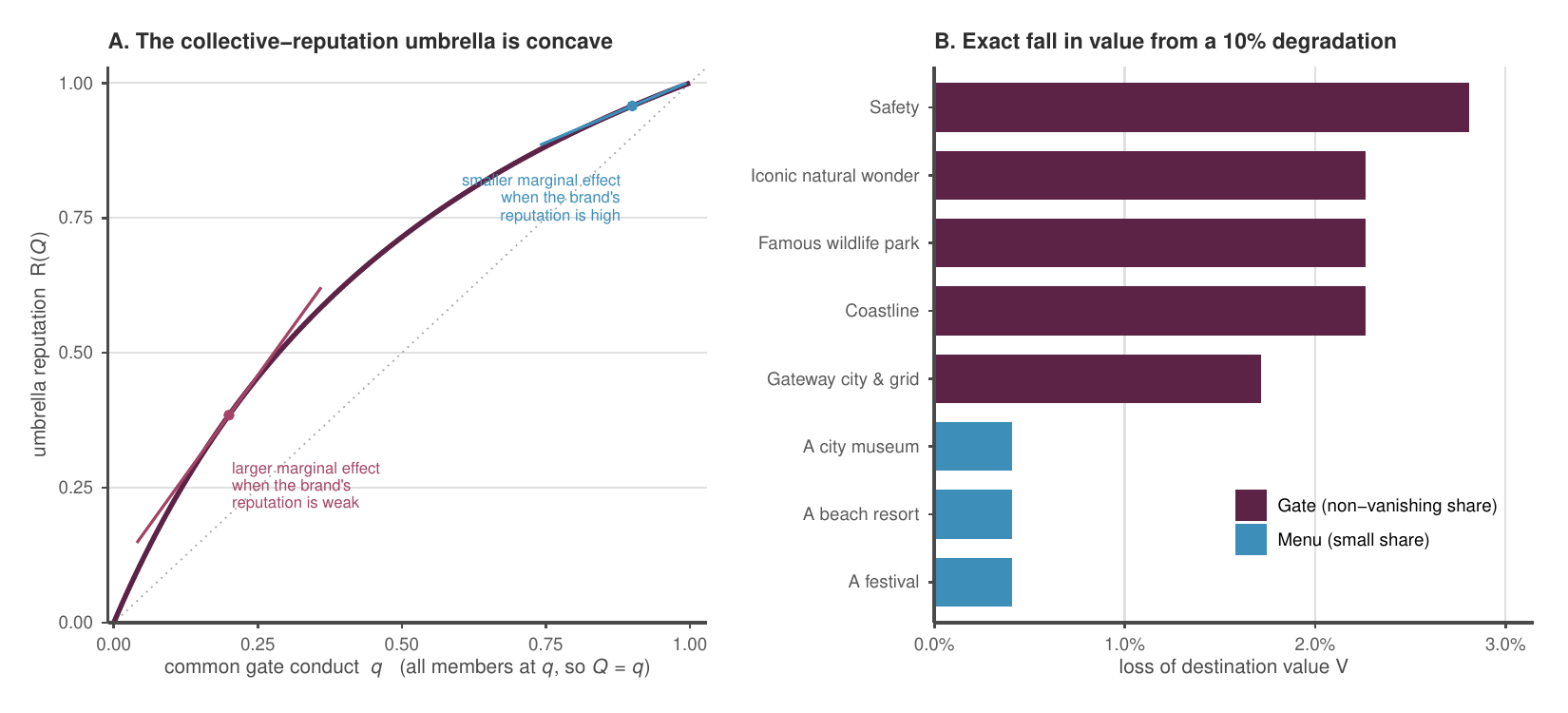}
\caption{The gate-and-menu asymmetry. \textbf{Left}: the collective-reputation umbrella
$R(Q)$ is concave, so the marginal effect of improving aggregate conduct $Q$ is
small when the brand's reputation is high and larger when it is weak
(Lemma~\ref{lem:concave}); the panel is drawn along the
common-conduct path, on which every member is at conduct $q$ and the weaker-link
statistic $Q$ equals $q$. \textbf{Right}: the exact proportional loss of destination
value from a 10\% quality degradation, using $\eta=-2$, $\rho=0.5$, and a regular
25-item menu. The assets are archetypes, not a specific destination; the parameters
are illustrative (Proposition~\ref{prop:asymmetry}).}
\label{fig:asymmetry}
\end{figure}

\section{Measurement: classifying gate and menu with a language model}\label{sec:measurement}

The model's constructs -- which assets are gate anchors, and how homogeneous a
country's gate is -- are latent and text-legible. This is the setting in which
\citet{ludwig2025llm} argue that large language models can be useful measurement
instruments, provided the measurement error is made explicit. We apply the
instrument to public guide text (the Wikivoyage country page and the Wikipedia
``Tourism in X'' and country articles) for every economy that meets a fixed
sample rule: non-missing World Bank tourism receipts and arrivals for at least
one year in 2015--2019, an existing Wikivoyage page and at least 1,500
characters of combined guide text. The rule yields 166 economies, spanning
brands from the single-theme (Maldives, Botswana) to the highly heterogeneous
(United States, India). Every excluded candidate is recorded, with its
exclusion reason, in the replication registry. For each country the model
returns a \textit{brand homogeneity} index
$H \in [0,100]$ -- the degree to which one theme dominates the brand -- and, for
each named attraction, a gate-versus-menu label.

Theme homogeneity is a proxy for
the model's quality-conduct homogeneity, not the same object: a country can have a
concentrated theme and still have low-quality members, or a heterogeneous theme mix
without quality contamination. The formal dispersion comparison additionally holds
mean conduct and conduct weights fixed; $H$ observes neither condition. It is
therefore a descriptive proxy rather than a direct estimate of $Q$.

Two language models from different families -- Claude (Opus) and OpenAI's Codex
(gpt-5.5) -- code every country independently from the same text and rubric,
through scripted, logged application-programming-interface calls (Appendix). They
agree strongly: the correlation of the homogeneity index across coders is $r =
0.84$ over the 166 economies (Figure~\ref{fig:reliability}), and they agree on
the dominant theme for 136 of 166 countries (82\%). The instrument is stable:
the scripted Claude coder reproduces the frozen original 42-country pass at
$r = 0.93$ on overlapping economies and reproduces itself at $r = 0.97$ in a
duplicate twenty-country run (mean absolute difference 2.3 points), while
rephrasing the rubric in two variants yields cross-variant correlations of
$0.59$--$0.89$ on a thirty-country subsample. This is evidence of reliability, not validity. The two coders use the same
source text and may share systematic training-corpus errors, so one model's
score is not a credible instrument for the other's error without a human
gold-standard sample -- which would require a large, internationally diverse
panel of raters and is beyond the scope of this paper. The empirical claims
below therefore do not use the model-against-model errors-in-variables
correction as identifying evidence; the measurement rests on cross-model
reliability and on the World Heritage odds ratio, an external check on the gate
labels.

The gate-versus-menu labels pass a partial external check. We match the named
attractions from each coder to the UNESCO World Heritage List within each country:
992 Claude-coded attractions, 685 Codex-coded attractions and 269 name-matched
attractions on which the coders agree. In the agreement sample, anchors appear on
the List at 6.4 times the odds of menu attractions (Fisher exact $p < 0.001$; the
odds ratios are 2.7 and 4.2 for each coder alone). Non-reproducible, iconic assets
sort into the gate, as part of the construct requires. This check bears on the
non-reproducibility and iconicity of the gate assets, not on the demand margin:
that gate assets move the decision to visit at all is built into the labeling
rubric, not tested by the World Heritage match. Direct evidence there would need
the reason-for-visit or booking data noted above, which do not exist at
cross-country scale.

\begin{figure}[htbp]
\centering
\includegraphics[width=0.7\linewidth]{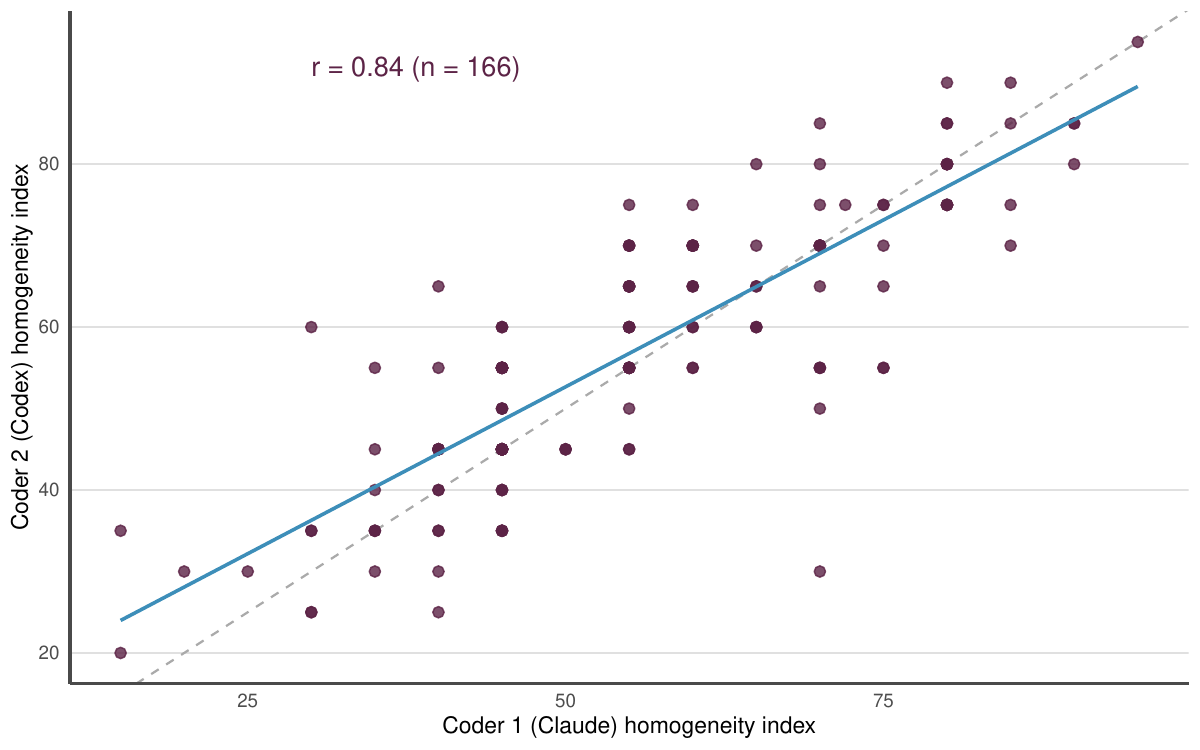}
\caption{Two language models from different families, coding brand homogeneity
independently from the same guide text, agree strongly ($r = 0.84$, $n = 166$
economies). Coder 1 is Claude (Opus); coder 2 is OpenAI Codex (gpt-5.5).}
\label{fig:reliability}
\end{figure}

\section{A within-country illustration of the contamination ordering}\label{sec:results}

Section~\ref{sec:measurement} showed that the labels are reliable and that their
iconicity and non-reproducibility component passes a partial external check.
Testing the model's conduct-contamination prediction (P2) is harder, because its
closest premium implication is inherently within-brand. Conditional on comparable
mean conduct and weights, a more homogeneous, better-maintained set has higher $Q$
and therefore higher $R(Q)$; relating this to a higher umbrella premium additionally
assumes that premiums increase with reputation. A
cross-country comparison cannot isolate it. Regressing log international tourism
receipts per arrival on the homogeneity index, log GDP per capita and region
fixed effects across the 164 economies with the required data gives a near-zero estimate
($\hat\beta = 0.0004$, $p = 0.94$, HC1 standard errors). The two economies lost from
the coded sample of 166 are the British Virgin Islands and Eritrea, for which the World
Bank reports no GDP per capita; both have the receipts and arrivals data. Receipts per arrival
confound price, length of stay, distance to source markets and reason for visit,
so the cross-country level cannot identify a within-brand mechanism. We therefore
illustrate the conditional implication at the level where it applies -- within one country --
using the four French wine regions of \citet{schatt2025collective}.

The theoretical construct is quality-conduct homogeneity, whereas the text index
measures thematic concentration and is only a proxy. If that proxy nevertheless
contains information about the coherence of the shared signal, it should rank
Bordeaux, the most heterogeneous of the four regions in the underlying study, low,
and it does. We code the four regions from their guide text with the same two-coder
procedure. The combined homogeneity index has a Spearman rank correlation
of $0.80$ with the estimated collective-reputation premiums on economic margin
(Table~\ref{tab:wine}). With $n=4$, the exact two-sided permutation $p$-value is
$0.33$, so this is not a powered test. It is an external illustration of a
conditional ordering, not a test of a model-derived premium sign. The clearest
result is at the end of the ranking: Bordeaux is the least homogeneous
region in both coders' scores and the only region with a negative estimated
umbrella premium.

\begin{table}[htbp]
\centering
\caption{Homogeneity scores for the four \citet{schatt2025collective} regions and
their estimated umbrella premiums. The mean homogeneity score has Spearman
$\rho = 0.80$ with the premium, but with $n = 4$ the exact two-sided permutation
$p$-value is $0.33$. Premium is their estimated coefficient on economic margin
(percentage points).}
\label{tab:wine}
\begin{tabular}{lcccc}
\toprule
Region & Claude $H$ & Codex $H$ & Mean $H$ & Estimated premium (pp) \\
\midrule
Champagne   & 88 & 90 & 89.0 & $+3.92$ \\
Burgundy    & 80 & 35 & 57.5 & $+4.19$ \\
Rh\^one     & 52 & 40 & 46.0 & $+0.36$ \\
Bordeaux    & 40 & 25 & 32.5 & $-2.72$ \\
\bottomrule
\end{tabular}
\end{table}

The largest discrepancy across coders indicates one limit of the measure: the coders disagree
most on Burgundy, which is uniformly prestigious but fragmented across small
producers, and fragmentation is not the construct. We use the unrefined
measure to avoid fitting it to the result. The exercise is an illustration, not
a validation of the instrument.

Extending this beyond four regions defines what a serious test of the premium
prediction would require. Such a test would code many regions within a single
country and relate their homogeneity to a demand-based outcome -- hotel or booking
revenue, or occupancy -- rather than to a published reputation premium. Such
premiums exist for only a small number of regions. A broader panel of collective
food-and-drink brands (geographical indications) is a second approach; we assembled
one such panel as a scoping exercise, and its materials are retained in the
replication package for researchers who wish to extend the measurement. The
contribution here is the method and the theory it measures, not a large-sample
conclusion.

\section{Implications for destination management and policy}\label{sec:policy}

The framework gives destination managers a three-stage decision protocol. First,
classify assets by the margin on which they enter demand: participation salience
and low reproducibility place an asset in the gate, while modular itinerary value
places it in the menu. Second, monitor gate quality and conduct to identify which
element is binding or deteriorating; visitor counts alone can miss the
cross-regional value at risk. Third, compare the welfare wedge in
\eqref{eq:wedge} with intervention cost and institutional feasibility, and assign
coordination to the lowest level of governance capable of internalizing the
spillover. This sequence does not require destination managers to place a monetary
value on every menu attraction before acting.

Three more specific implications follow.

First, under the conditions of Proposition~\ref{prop:commons} and holding the other
qualities and conduct fixed, conservation of a gate asset is weakly undersupplied.
The first-best
object is the domestic welfare wedge $\tau_k^{W}$ in \eqref{eq:wedge}, not gross
tourism receipts. A conservation subsidy, development tax or club-transfer schedule
should have marginal slope equal to the domestic surplus that degradation would cost
other regions, adjusted for financing costs and incidence; a per-arrival revenue
share must be converted using the local arrival derivative. The model therefore
supports targeted assessment of anchors and essential inputs rather than an
automatic uniform allocation across attractions; actual spending priorities must
also reflect the cost and tractability of improving each gate element. Tourism
taxes are increasingly analyzed as corrective instruments of this kind
\citep{rosellonadal2026taxation}.

Second, Propositions~\ref{prop:weakerlink} and~\ref{prop:commons} direct diagnostic attention to the
lowest weighted-quality anchor or the input most vulnerable to degradation, rather than
automatically to the most visible site. Whether that asset receives the
highest-return spending depends additionally on intervention costs. This is an argument for
coordination, not automatically for centralization. A national brand authority, an
appellation-style club or a fiscal transfer system can each implement the wedge
only if it has enough information and discipline to target the binding constraint
rather than the politically visible one. Each also requires that capture does not
prevent enforcement against the relevant custodian.

Third, the reputational gate requires the standard collective-reputation
instruments \citep{nie2024collective, winfreemccluskey2005}: minimum quality
standards and monitoring at the anchors, and credible enforcement against low-conduct
members whose behavior, through $Q$ and Lemma~\ref{lem:concave}, damages the
umbrella for everyone. The institutional constraint is credibility: a captured
regulator will not discipline the connected custodian who binds the weak link. Anti-counterfeiting
and quality enforcement, which
\citet{schatt2025collective} list among their policy levers, are in this model
investments in the shared gate.

\section{Conclusion}\label{sec:conclusion}

This paper asked two questions: when does a collective tourism brand protect
its members, and what is a single attraction worth within it? Both answers
follow from one structure. A collective tourism brand has two margins. It is a
gate of scarce,
non-reproducible anchors and essential inputs that act directly on participation,
and it is a menu of reproducible attractions that primarily governs time allocation
and reaches participation through an inclusive-value share. Because the gate is
weaker-link and multiplies the menu, a
negative shock to one gate asset affects the whole brand. Because the menu aggregates
many modular attractions, one more menu attraction is a small share of the brand.

Three further results give the two answers. On protection: the reputation attached to
the gate obeys the same weakest-member logic as its physical quality, so one
low-conduct member reduces the shared reputation for everyone, disproportionately
so when the brand's standing is already low. On valuation: conditional on fixed
conduct, Aumann--Shapley accounting gives each physical asset a value share that
does not require a discretionary restoration floor --
non-vanishing for a binding gate anchor, small for any one menu attraction. And on why
protection fails: the gate is a commons, because the custodian of a gate asset
does not capture the domestic surplus its maintenance creates for other regions.
Under the conditions in Proposition~\ref{prop:commons}, its conditional physical-quality
choice is therefore weakly below the internalizing choice.

The measurement exercise is deliberately modest and is offered as a method, not a
conclusion. Two language-model coders produce similar homogeneity scores for 166
economies, and the attractions they jointly label anchors appear on the UNESCO World
Heritage List at several times the odds of menu attractions, so the labels are
reliable and pass a partial external check on iconicity and non-reproducibility.
The contamination ordering is inherently within-brand
-- cross-country receipts per arrival are too confounded to identify it -- and we
illustrate a conditional premium implication on four French wine regions, where the
index places Bordeaux at the lowest position. This is not a powered test or a
model-derived premium sign. The main contribution
is the framework: it identifies which assets determine the shared brand, how their
value shares should be accounted for, and the conditions under which decentralized
custodians weakly under-maintain physical quality. For destination management, it
separates three decisions that are often combined: which assets belong in the
participation gate, which gate element currently constrains the brand, and which
institution can maintain it at least cost. That separation helps policy target
cross-regional brand value rather than visibility alone.

\bibliographystyle{plainnat}
\bibliography{references}

@article{kremer1993oring,
  author  = {Kremer, Michael},
  title   = {The {O}-Ring Theory of Economic Development},
  journal = {Quarterly Journal of Economics},
  year    = {1993},
  volume  = {108},
  number  = {3},
  pages   = {551--575}
}

@techreport{ludwig2025llm,
  author      = {Ludwig, Jens and Mullainathan, Sendhil and Rambachan, Ashesh},
  title       = {Large Language Models: An Applied Econometric Framework},
  institution = {National Bureau of Economic Research},
  type        = {NBER Working Paper},
  number      = {33344},
  year        = {2025},
  note        = {Forthcoming, Annual Review of Economics}
}

@article{tirole1996collective,
  author  = {Tirole, Jean},
  title   = {A Theory of Collective Reputations (with Applications to the Persistence of Corruption and to Firm Quality)},
  journal = {Review of Economic Studies},
  year    = {1996},
  volume  = {63},
  number  = {1},
  pages   = {1--22}
}

@article{schatt2025collective,
  author  = {Schatt, Alain and Weisskopf, Jean-Philippe},
  title   = {Unveiling the Collective Reputation Effect of {F}rench Wines},
  journal = {Tourism Economics},
  year    = {2025},
  volume  = {31},
  number  = {6},
  pages   = {1140--1159},
  doi     = {10.1177/13548166251332340}
}

@article{nie2024collective,
  author  = {Nie, Kun-xi},
  title   = {Collective Reputation of Tourism Attraction: A Case Study of {L}ijiang City in {C}hina},
  journal = {Tourism Planning and Development},
  year    = {2024},
  volume  = {21},
  number  = {3},
  pages   = {363--373},
  doi     = {10.1080/21568316.2022.2089222}
}

@article{corneshartley2007,
  author  = {Cornes, Richard and Hartley, Roger},
  title   = {Weak Links, Good Shots and Other Public Good Games: Building on {BBV}},
  journal = {Journal of Public Economics},
  year    = {2007},
  volume  = {91},
  number  = {9},
  pages   = {1684--1707}
}

@article{winfreemccluskey2005,
  author  = {Winfree, Jason A. and McCluskey, Jill J.},
  title   = {Collective Reputation and Quality},
  journal = {American Journal of Agricultural Economics},
  year    = {2005},
  volume  = {87},
  number  = {1},
  pages   = {206--213}
}

@article{nockestrausz2023,
  author  = {Nocke, Volker and Strausz, Roland},
  title   = {Collective Brand Reputation},
  journal = {Journal of Political Economy},
  year    = {2023},
  volume  = {131},
  number  = {1},
  pages   = {1--58}
}

@article{costanigro2010nested,
  author  = {Costanigro, Marco and McCluskey, Jill J. and Goemans, Christopher},
  title   = {The Economics of Nested Names: Name Specificity, Reputations, and Price Premia},
  journal = {American Journal of Agricultural Economics},
  year    = {2010},
  volume  = {92},
  number  = {5},
  pages   = {1339--1350}
}

@article{castriotadelmastro2015,
  author  = {Castriota, Stefano and Delmastro, Marco},
  title   = {The Economics of Collective Reputation: Evidence from the Wine Industry},
  journal = {American Journal of Agricultural Economics},
  year    = {2015},
  volume  = {97},
  number  = {2},
  pages   = {469--489}
}

@article{shapiro1983premiums,
  author  = {Shapiro, Carl},
  title   = {Premiums for High Quality Products as Returns to Reputations},
  journal = {Quarterly Journal of Economics},
  year    = {1983},
  volume  = {98},
  number  = {4},
  pages   = {659--680}
}

@article{rosen1981superstars,
  author  = {Rosen, Sherwin},
  title   = {The Economics of Superstars},
  journal = {American Economic Review},
  year    = {1981},
  volume  = {71},
  number  = {5},
  pages   = {845--858}
}

@article{mussarosen1978,
  author  = {Mussa, Michael and Rosen, Sherwin},
  title   = {Monopoly and Product Quality},
  journal = {Journal of Economic Theory},
  year    = {1978},
  volume  = {18},
  number  = {2},
  pages   = {301--317}
}

@article{candelafigini2010,
  author  = {Candela, Guido and Figini, Paolo},
  title   = {Destination Unknown. Is There Any Economics Beyond Tourism Areas?},
  journal = {Review of Economic Analysis},
  year    = {2010},
  volume  = {2},
  number  = {3},
  pages   = {256--271}
}

@book{candelafigini2012,
  author    = {Candela, Guido and Figini, Paolo},
  title     = {The Economics of Tourism Destinations},
  publisher = {Springer},
  address   = {Berlin and Heidelberg},
  year      = {2012}
}

@article{gergaud2017netbenefits,
  author  = {Gergaud, Olivier and Livat, Florine and Rickard, Bradley and Warzynski, Frederic},
  title   = {Evaluating the Net Benefits of Collective Reputation: The Case of {B}ordeaux Wine},
  journal = {Food Policy},
  year    = {2017},
  volume  = {71},
  pages   = {8--16}
}

@article{hirshleifer1983weakest,
  author  = {Hirshleifer, Jack},
  title   = {From Weakest-Link to Best-Shot: The Voluntary Provision of Public Goods},
  journal = {Public Choice},
  year    = {1983},
  volume  = {41},
  number  = {3},
  pages   = {371--386}
}

@article{cornes1993dyke,
  author  = {Cornes, Richard},
  title   = {Dyke Maintenance and Other Stories: Some Neglected Types of Public Goods},
  journal = {Quarterly Journal of Economics},
  year    = {1993},
  volume  = {108},
  number  = {1},
  pages   = {259--271}
}

@book{cornes1996theory,
  author    = {Cornes, Richard and Sandler, Todd},
  title     = {The Theory of Externalities, Public Goods, and Club Goods},
  edition   = {2},
  publisher = {Cambridge University Press},
  address   = {Cambridge},
  year      = {1996}
}

@article{fabergaubert2019,
  author  = {Faber, Benjamin and Gaubert, C{\'e}cile},
  title   = {Tourism and Economic Development: Evidence from {M}exico's Coastline},
  journal = {American Economic Review},
  year    = {2019},
  volume  = {109},
  number  = {6},
  pages   = {2245--2293}
}

@article{magontier2024coastal,
  author  = {Magontier, Pierre and Sol{\'e}-Oll{\'e}, Albert and Viladecans-Marsal, Elisabet},
  title   = {The Political Economy of Coastal Development},
  journal = {Journal of Public Economics},
  year    = {2024},
  volume  = {238},
  pages   = {105178}
}

@article{balboni2025harm,
  author  = {Balboni, Clare},
  title   = {In Harm's Way? Infrastructure Investments and the Persistence of Coastal Cities},
  journal = {American Economic Review},
  year    = {2025},
  volume  = {115},
  number  = {1},
  pages   = {77--116}
}

@article{mukanjari2021valuation,
  author  = {Mukanjari, Samson and Ntuli, Herbert and Muchapondwa, Edwin},
  title   = {Valuation of Nature-Based Tourism Using Contingent Valuation Survey: Evidence from {S}outh {A}frica},
  journal = {Journal of Environmental Economics and Policy},
  year    = {2021},
  doi     = {10.1080/21606544.2021.2010604}
}

@book{aumannshapley1974,
  author    = {Aumann, Robert J. and Shapley, Lloyd S.},
  title     = {Values of Non-Atomic Games},
  publisher = {Princeton University Press},
  address   = {Princeton},
  year      = {1974}
}

@article{dixitstiglitz1977,
  author  = {Dixit, Avinash K. and Stiglitz, Joseph E.},
  title   = {Monopolistic Competition and Optimum Product Diversity},
  journal = {American Economic Review},
  year    = {1977},
  volume  = {67},
  number  = {3},
  pages   = {297--308}
}

@article{wernerfelt1988umbrella,
  author  = {Wernerfelt, Birger},
  title   = {Umbrella Branding as a Signal of New Product Quality: An Example of Signalling by Posting a Bond},
  journal = {RAND Journal of Economics},
  year    = {1988},
  volume  = {19},
  number  = {3},
  pages   = {458--466}
}

@article{sandmo1975externalities,
  author  = {Sandmo, Agnar},
  title   = {Optimal Taxation in the Presence of Externalities},
  journal = {Swedish Journal of Economics},
  year    = {1975},
  volume  = {77},
  number  = {1},
  pages   = {86--98}
}

@article{atkinsonstern1974,
  author  = {Atkinson, Anthony B. and Stern, Nicholas H.},
  title   = {Pigou, Taxation and Public Goods},
  journal = {Review of Economic Studies},
  year    = {1974},
  volume  = {41},
  number  = {1},
  pages   = {119--128}
}

@article{ballardfullerton1992,
  author  = {Ballard, Charles L. and Fullerton, Don},
  title   = {Distortionary Taxes and the Provision of Public Goods},
  journal = {Journal of Economic Perspectives},
  year    = {1992},
  volume  = {6},
  number  = {3},
  pages   = {117--131}
}

@article{diamond1973externalities,
  author  = {Diamond, Peter A.},
  title   = {Consumption Externalities and Imperfect Corrective Pricing},
  journal = {The Bell Journal of Economics and Management Science},
  year    = {1973},
  volume  = {4},
  number  = {2},
  pages   = {526--538}
}

@article{clarke1971multipart,
  author  = {Clarke, Edward H.},
  title   = {Multipart Pricing of Public Goods},
  journal = {Public Choice},
  year    = {1971},
  volume  = {11},
  pages   = {17--33}
}

@article{grovesledyard1977,
  author  = {Groves, Theodore and Ledyard, John},
  title   = {Optimal Allocation of Public Goods: A Solution to the ``Free Rider'' Problem},
  journal = {Econometrica},
  year    = {1977},
  volume  = {45},
  number  = {4},
  pages   = {783--809}
}

@book{ostrom1990governing,
  author    = {Ostrom, Elinor},
  title     = {Governing the Commons: The Evolution of Institutions for Collective Action},
  publisher = {Cambridge University Press},
  address   = {Cambridge},
  year      = {1990}
}

@article{anholt2007brand,
  author  = {Anholt, Simon},
  title   = {Competitive Identity: The New Brand Management for Nations, Cities and Regions},
  journal = {Journal of Brand Management},
  year    = {2007},
  volume  = {14},
  number  = {6},
  pages   = {474--475}
}

@article{krutilla1967conservation,
  author  = {Krutilla, John V.},
  title   = {Conservation Reconsidered},
  journal = {American Economic Review},
  year    = {1967},
  volume  = {57},
  number  = {4},
  pages   = {777--786}
}

@book{maler1974environmental,
  author    = {M{\"a}ler, Karl-G{\"o}ran},
  title     = {Environmental Economics: A Theoretical Inquiry},
  publisher = {Johns Hopkins University Press for Resources for the Future},
  address   = {Baltimore},
  year      = {1974}
}

@book{freeman2014measurement,
  author    = {Freeman, A. Myrick, III and Herriges, Joseph A. and Kling, Catherine L.},
  title     = {The Measurement of Environmental and Resource Values: Theory and Methods},
  edition   = {3},
  publisher = {RFF Press / Routledge},
  address   = {New York},
  year      = {2014}
}

@incollection{carsonhanemann2005cv,
  author    = {Carson, Richard T. and Hanemann, W. Michael},
  title     = {Contingent Valuation},
  booktitle = {Handbook of Environmental Economics, Volume 2},
  editor    = {M{\"a}ler, Karl-G{\"o}ran and Vincent, Jeffrey R.},
  publisher = {Elsevier},
  address   = {Amsterdam},
  year      = {2005},
  pages     = {821--936}
}

@incollection{phaneufsmith2005recreation,
  author    = {Phaneuf, Daniel J. and Smith, V. Kerry},
  title     = {Recreation Demand Models},
  booktitle = {Handbook of Environmental Economics, Volume 2},
  editor    = {M{\"a}ler, Karl-G{\"o}ran and Vincent, Jeffrey R.},
  publisher = {Elsevier},
  address   = {Amsterdam},
  year      = {2005},
  pages     = {671--761}
}

@article{adamowicz1994combining,
  author  = {Adamowicz, Wiktor and Louviere, Jordan and Williams, Michael},
  title   = {Combining Revealed and Stated Preference Methods for Valuing Environmental Amenities},
  journal = {Journal of Environmental Economics and Management},
  year    = {1994},
  volume  = {26},
  number  = {3},
  pages   = {271--292}
}

@article{costanza1997value,
  author  = {Costanza, Robert and d'Arge, Ralph and de Groot, Rudolf and Farber, Stephen and Grasso, Monica and Hannon, Bruce and Limburg, Karin and Naeem, Shahid and O'Neill, Robert V. and Paruelo, Jose and Raskin, Robert G. and Sutton, Paul and van den Belt, Marjan},
  title   = {The Value of the World's Ecosystem Services and Natural Capital},
  journal = {Nature},
  year    = {1997},
  volume  = {387},
  pages   = {253--260}
}

@article{herriges2004whats,
  author  = {Herriges, Joseph A. and Kling, Catherine L. and Phaneuf, Daniel J.},
  title   = {What's the Use? Welfare Estimates from Revealed Preference Models When Weak Complementarity Does Not Hold},
  journal = {Journal of Environmental Economics and Management},
  year    = {2004},
  volume  = {47},
  number  = {1},
  pages   = {55--70}
}

@article{naidoo2005biodiversity,
  author  = {Naidoo, Robin and Adamowicz, Wiktor L.},
  title   = {Biodiversity and Nature-Based Tourism at Forest Reserves in {U}ganda},
  journal = {Environment and Development Economics},
  year    = {2005},
  volume  = {10},
  number  = {2},
  pages   = {159--178}
}

@article{konishisandfort2003,
  author  = {Konishi, Hideo and Sandfort, Michael T.},
  title   = {Anchor Stores},
  journal = {Journal of Urban Economics},
  year    = {2003},
  volume  = {53},
  number  = {3},
  pages   = {413--435},
  doi     = {10.1016/S0094-1190(03)00002-0}
}

@article{gouldpashigian2005,
  author  = {Gould, Eric D. and Pashigian, B. Peter and Prendergast, Canice J.},
  title   = {Contracts, Externalities, and Incentives in Shopping Malls},
  journal = {Review of Economics and Statistics},
  year    = {2005},
  volume  = {87},
  number  = {3},
  pages   = {411--422},
  doi     = {10.1162/0034653054638355}
}

@article{cabral2009umbrella,
  author  = {Cabral, Lu{\'\i}s M. B.},
  title   = {Umbrella Branding with Imperfect Observability and Moral Hazard},
  journal = {International Journal of Industrial Organization},
  year    = {2009},
  volume  = {27},
  number  = {2},
  pages   = {206--213},
  doi     = {10.1016/j.ijindorg.2008.07.002}
}

@article{gentzkow2019text,
  author  = {Gentzkow, Matthew and Kelly, Bryan and Taddy, Matt},
  title   = {Text as Data},
  journal = {Journal of Economic Literature},
  year    = {2019},
  volume  = {57},
  number  = {3},
  pages   = {535--574},
  doi     = {10.1257/jel.20181020}
}

@article{pintassilgosilva2007,
  author  = {Pintassilgo, Pedro and Silva, Jo{\~a}o Albino},
  title   = {`Tragedy of the Commons' in the Tourism Accommodation Industry},
  journal = {Tourism Economics},
  year    = {2007},
  volume  = {13},
  number  = {2},
  pages   = {209--224},
  doi     = {10.5367/000000007780823168}
}

@article{crouchritchie1999,
  author  = {Crouch, Geoffrey I. and Ritchie, J. R. Brent},
  title   = {Tourism, Competitiveness, and Societal Prosperity},
  journal = {Journal of Business Research},
  year    = {1999},
  volume  = {44},
  number  = {3},
  pages   = {137--152},
  doi     = {10.1016/S0148-2963(97)00196-3}
}

@article{bacchiega2024dilution,
  author  = {Bacchiega, Emanuele and Colucci, Mariachiara and Denicol{\`o}, Vincenzo and Magnani, Marco},
  title   = {Only the Ugly Face? A Theoretical Model of Brand Dilution},
  journal = {Management Science},
  year    = {2024},
  volume  = {70},
  number  = {5},
  pages   = {3182--3199},
  doi     = {10.1287/mnsc.2022.00852}
}

@article{choi2010heritage,
  author  = {Choi, Andy S. and Ritchie, Brent W. and Papandrea, Franco and Bennett, Jeff},
  title   = {Economic Valuation of Cultural Heritage Sites: A Choice Modeling Approach},
  journal = {Tourism Management},
  year    = {2010},
  volume  = {31},
  number  = {2},
  pages   = {213--220},
  doi     = {10.1016/j.tourman.2009.02.014}
}

@article{dell2025deep,
  author  = {Dell, Melissa},
  title   = {Deep Learning for Economists},
  journal = {Journal of Economic Literature},
  year    = {2025},
  volume  = {63},
  number  = {1},
  pages   = {5--58},
  doi     = {10.1257/jel.20241733}
}

@article{rosellonadal2026taxation,
  author  = {Rossell{\'o}-Nadal, Jaume and Sard, Mar{\'\i}a},
  title   = {Tourism Taxation: Balancing Revenues, Competitiveness and Sustainability in Destination Management},
  journal = {Tourism Management},
  year    = {2026},
  volume  = {113},
  pages   = {105326},
  doi     = {10.1016/j.tourman.2025.105326}
}

@article{jiao2025spillover,
  author  = {Jiao, Xiaoying and Chen, Jason Li and Li, Gang and Liu, Tongxiang},
  title   = {From Tourism Demand to Destination Competitiveness: A Spatial Spillover Perspective},
  journal = {Tourism Management},
  year    = {2025},
  volume  = {111},
  pages   = {105236},
  doi     = {10.1016/j.tourman.2025.105236}
}

@article{kotlergertner2002,
  author  = {Kotler, Philip and Gertner, David},
  title   = {Country as Brand, Product, and Beyond: A Place Marketing and Brand Management Perspective},
  journal = {Journal of Brand Management},
  year    = {2002},
  volume  = {9},
  number  = {4--5},
  pages   = {249--261},
  doi     = {10.1057/palgrave.bm.2540076}
}

@article{swain2024place,
  author  = {Swain, Swapnarag and Jebarajakirthy, Charles and Sharma, Bhuvanesh Kumar and Maseeh, Haroon Iqbal and Agrawal, Amee and Shah, Jinal and Saha, Raiswa},
  title   = {Place Branding: A Systematic Literature Review and Future Research Agenda},
  journal = {Journal of Travel Research},
  year    = {2024},
  volume  = {63},
  number  = {3},
  pages   = {535--564},
  doi     = {10.1177/00472875231168620}
}

@article{chisik2003reputational,
  author  = {Chisik, Richard},
  title   = {Export Industry Policy and Reputational Comparative Advantage},
  journal = {Journal of International Economics},
  year    = {2003},
  volume  = {59},
  number  = {2},
  pages   = {423--451},
  doi     = {10.1016/S0022-1996(02)00020-X}
}

@article{choi1998brand,
  author  = {Choi, Jay Pil},
  title   = {Brand Extension as Informational Leverage},
  journal = {Review of Economic Studies},
  year    = {1998},
  volume  = {65},
  number  = {4},
  pages   = {655--669},
  doi     = {10.1111/1467-937X.00063}
}

\newpage
\appendix

\section*{Appendix: Measurement reproducibility}
\addcontentsline{toc}{section}{Appendix: Measurement reproducibility}

\paragraph{Sample rule and registry.} The country sample is defined by a rule,
not a list: World Bank tourism receipts and arrivals non-missing for at least
one year in 2015--2019, an existing Wikivoyage page, and at least 1,500
characters of combined guide text across the three sources. The registry file
(\texttt{country\_registry.csv}) retains every candidate World Bank economy
with its probe results and, where excluded, the exclusion reason.

\paragraph{Coder 1 (Claude).} Coder 1 runs as scripted batch calls to the
Anthropic API. The script is \texttt{15\_classify\_claude.py}: model pinned to
\texttt{claude-opus-4-8}, one request per country, extended
reasoning disabled and the response constrained to a fixed JSON schema. Every
call is logged in \texttt{coder1\_log.jsonl} with the batch and request
identifiers, a timestamp, the exact model identifier echoed by the interface,
SHA-256 hashes of the rendered prompt and of the rubric file, token counts
and the verbatim response. Current hosted models do not accept a sampling
temperature and do not guarantee byte-identical outputs, so instead of a
determinism claim we quantify run-to-run variation directly: a duplicate run
on a random twenty-country subsample is reported alongside the main results.
The original 42-country coder-1 scores (\texttt{coder1\_claude.csv}), which
were produced interactively without logged settings, are retained as an
archival file; the scripted coder reproduces them at $r = 0.93$ on the
overlapping economies, and all headline numbers use the scripted run.

\paragraph{Coder 2 (Codex).} Coder 2 runs through the OpenAI Codex
command-line interface in chunks of about fifteen countries
(\texttt{16\_classify\_codex\_full.py}), model \texttt{gpt-5.5}. The rubric is
repeated in full in every chunk; chunk membership is randomized with a fixed
seed and recorded, with the rubric's SHA-256 hash and the run parameters, in a
build manifest; every generated prompt and raw transcript is retained; and the
parser refuses to write output unless every registry country appears exactly
once across chunks. The audit trail for the full 166-economy run is thus
complete in its inputs -- prompts, rubric hash and chunk assignment -- but the
run pins less of the model environment than coder 1 does: the Codex
command-line version and any model-side sampling settings for the full run were
not separately logged, and the interface exposes no temperature or
reasoning-effort control to pin. We therefore treat coder 2 as a second,
imperfectly reproducible reading and rely on cross-coder reliability rather than
on either coder's byte-level determinism. The original 42-country transcript
(\texttt{codex\_raw.txt}, Codex CLI v0.142.2) is retained as an archival file.

\paragraph{Rubrics and prompts.} The canonical rubric (\texttt{rubric.md},
duplicated as \texttt{rubrics/rubric\_A.md}) and the two rephrased variants used
in the prompt-robustness check (\texttt{rubric\_B.md}, \texttt{rubric\_C.md})
are frozen in the replication package, together with every generated prompt, raw
transcript and the deterministic analysis scripts. The scoping materials for the
collective-brand (geographical-indication) extension noted in
Section~\ref{sec:results} -- \texttt{gi\_premia.csv}, its rubric variant and the
analysis script -- are retained there too, for researchers who wish to build on
them.

\section*{Appendix: Proofs}
\addcontentsline{toc}{section}{Appendix: Proofs}

\paragraph{Notation.} $\gamma_k = \phi_k g_k^{\eta}/\sum_l \phi_l g_l^{\eta}$ is the
gate share, with $\partial \ln G/\partial \ln g_k = \gamma_k$; $s_i = w_i
a_i^{\rho}/\sum_l w_l a_l^{\rho}$ is the menu share, with $\partial \ln A/\partial
\ln a_i = s_i$; $\psi = \partial \ln D/\partial \ln V_R$ once reputation is included.

\begin{proof}[Proof of Proposition~\ref{prop:asymmetry}]
Since $V = G(g)A(a)$ with $G$ depending only on $g$ and $A$ only on $a$,
$\partial \ln V/\partial \ln g_k = \partial \ln G/\partial \ln g_k = \gamma_k$ and
$\partial \ln V/\partial \ln a_i = \partial \ln A/\partial \ln a_i = s_i$, which is
\eqref{eq:elasticities}. (i) For the weakest gate asset, $(g_l/g_k)^{\eta} \le 1$
for every $l$ (since $g_l \ge g_k$ and $\eta < 0$), so $\gamma_k =
\phi_k/\sum_l \phi_l (g_l/g_k)^{\eta} \ge \phi_k > 0$, a bound independent of the
menu; as $\eta \to -\infty$, $\gamma_k \to 1$ for a unique binding asset
(Proposition \ref{prop:weakerlink}). Under menu regularity, $c \le w_i a_i^{\rho}
\le C$ for constants $0 < c \le C$, so $s_i \le C/(|M|c) = O(1/|M|)$. Hence
$\gamma_k/s_i \ge \phi_k |M| c/C \to \infty$ as the menu grows. For a fixed finite
proportional degradation $\lambda\in(0,1)$, the exact ratios are
\[
 \frac{G(\lambda g_k,g_{-k})}{G(g)}
 =\left[1+\gamma_k(\lambda^\eta-1)\right]^{1/\eta},\qquad
 \frac{A(\lambda a_i,a_{-i})}{A(a)}
 =\left[1+s_i(\lambda^\rho-1)\right]^{1/\rho}.
\]
The menu loss tends to zero as $s_i\to0$, while the weakest gate loss remains
positive for fixed $\gamma_k\ge\phi_k>0$. Thus the same asymptotic ordering holds,
although the finite loss ratio is not exactly $\gamma_k/s_i$. (ii) The cross
partials are
$\partial^2 V/\partial g_k\partial a_i = (\partial G/\partial g_k)(\partial A/
\partial a_i) > 0$ for all $i$ since $\partial G/\partial g_k = G\gamma_k/g_k > 0$
and $\partial A/\partial a_i = As_i/a_i > 0$; and, for $l \ne k$, $\partial^2 V/\partial g_k\partial
g_l = A\,\partial^2 G/\partial g_k\partial g_l > 0$ by the CES cross derivative
$\partial^2 G/\partial g_k\partial g_l = (1-\eta)G\gamma_k\gamma_l/(g_kg_l) > 0$
(since $\eta < 0 < 1$); see the proof of Proposition~\ref{prop:aggregate}. For the
scaling claims, work in elasticities: $\partial V/\partial a_j = G \cdot
\partial A/\partial a_j$, so $\partial \ln(\partial V/\partial a_j)/\partial \ln
g_k = \partial \ln G/\partial \ln g_k = \gamma_k$ for every $j \in M$, and for
$l \ne k$, $\partial \ln(\partial V/\partial g_l)/\partial \ln g_k =
(1-\eta)\gamma_k \ge \gamma_k$. In the other direction a change in $a_i$ leaves
$G$, the gate shares, and hence every elasticity $\partial \ln V/\partial \ln g_k
= \gamma_k$ unchanged; it moves the level $\partial V/\partial g_k = A\,\partial
G/\partial g_k$ only through the common scale factor $A$, with
$\partial \ln(\partial V/\partial g_k)/\partial \ln a_i = s_i$, which vanishes as
the menu grows. (iii) Write the
menu sum as $U_A = \sum_{i \in M} w_i a_i^{\rho} = A^{\rho}$. Adding an attraction
with weight $w$ and quality $a$ raises the menu to $\tilde A = (U_A + w
a^{\rho})^{1/\rho}$, so the proportional gain in destination value is
\[
  \frac{\Delta V}{V} = \frac{\tilde A - A}{A}
  = \left( 1 + \frac{w a^{\rho}}{U_A} \right)^{1/\rho} - 1
  \approx \frac{1}{\rho}\,\frac{w a^{\rho}}{U_A},
\]
using the first-order expansion $(1+\varepsilon)^{1/\rho} - 1 \approx
\varepsilon/\rho$ for small $\varepsilon = w a^{\rho}/U_A$. Under menu regularity
the incumbent terms satisfy $w_i a_i^{\rho} \ge c > 0$, so $U_A \ge |M|c \to
\infty$ as the menu grows; with the entrant's term $w a^{\rho} \le C$ bounded,
$w a^{\rho}/U_A \to 0$ and $\Delta V/V \to 0$. Note that the absolute increment $\tilde A - A$ need not
fall, since $A$ itself grows with the menu; it is the share of any one
attraction in brand value that vanishes, consistent with the menu-share elasticity
$s_i = O(1/|M|)$ in part~(i).
\end{proof}

\begin{proof}[Proof of Proposition~\ref{prop:weakerlink}]
From \eqref{eq:gateshare}, $\partial \ln \gamma_k/\partial \ln g_k = \eta(1 -
\gamma_k)$. With $|K| \ge 2$ every share satisfies $\gamma_k < 1$, so for $\eta < 0$
this is strictly negative and $\gamma_k$ is strictly decreasing in $g_k$. The
ranking claim follows from $\gamma_k/\gamma_l = (\phi_k/\phi_l)(g_k/g_l)^{\eta}$:
shares order assets by weighted quality $\phi_k g_k^{\eta}$, and with equal
weights $(g_k/g_l)^{\eta} > 1$ iff $g_k < g_l$ (as $\eta<0$), so the smallest
$g_k$ has the largest share. As $\eta \to -\infty$, let $m = \min_l g_l$
and $I = \{l : g_l = m\}$, and divide numerator and denominator of $\gamma_k =
\phi_k (g_k/m)^{\eta}/\sum_l \phi_l (g_l/m)^{\eta}$ by $m^{\eta}$: each term
$(g_l/m)^{\eta} \to 0$ for $g_l > m$ (since $g_l/m > 1$ and $\eta \to -\infty$) and
$(g_l/m)^{\eta} = 1$ for $l \in I$. Hence $\gamma_k \to \phi_k/\sum_{l \in I}\phi_l$
for $k \in I$ (which is $1$ when $I = \{k\}$) and $\gamma_k \to 0$ for $k \notin I$,
while $G = m[\sum_l \phi_l (g_l/m)^{\eta}]^{1/\eta} \to m$, because $\sum_l \phi_l
(g_l/m)^{\eta} \to \sum_{l \in I}\phi_l \in (0,1]$ and any positive constant raised
to the power $1/\eta \to 0$ tends to $1$; thus $G \to \min_l g_l$. The menu set $M$
does not enter $G$ by construction.
\end{proof}

\begin{proof}[Proof of Lemma~\ref{lem:concave}]
Write $m = 1 - x \in (0,1)$, so $R(Q) = Q/[m + (1-m)Q]$. Then
$R'(Q) = m/[m + (1-m)Q]^2 > 0$ and
$R''(Q) = -2m(1-m)/[m + (1-m)Q]^3 < 0$, so $R$ is strictly increasing and strictly
concave. $R(0) = 0$ and $R(1) = 1$. Evaluating, $R'(0) = 1/m = 1/(1-x)$ and
$R'(1) = m = 1-x$. Expressed in $x$,
$R'(Q) = (1-x)/[(1-x) + xQ]^2$.
Since $\partial Q/\partial q_k=Q\chi_k/q_k$,
\[
 \frac{\partial\ln R(Q(q))}{\partial\ln q_k}
 =\frac{QR'(Q)}{R(Q)}\chi_k
 =\frac{1-x}{(1-x)+xQ}\chi_k,
\]
which proves \eqref{eq:conductelasticity}. With equal $\omega_k$, the CES $Q$ is
symmetric and concave for $\eta<1$, hence Schur-concave; at a fixed arithmetic mean,
a mean-preserving spread weakly lowers $Q$ (strictly away from permutations).
\end{proof}

\begin{proof}[Proof of Proposition~\ref{prop:aggregate}]
Fix $q$. Then $\ln T_R = \ln \bar M + \ln D(V_R)$ and
$\partial\ln V_R/\partial\ln g_k=\partial\ln V/\partial\ln g_k=\gamma_k$, so
$\partial \ln T_R/\partial \ln g_k = \psi\gamma_k > 0$. The CES gate cross derivative
is, differentiating $\partial G/\partial g_k = U^{1/\eta - 1}\phi_k g_k^{\eta -
1}$ (where $G = U^{1/\eta}$, $U = \sum_l \phi_l g_l^{\eta}$) and using $\phi_k
g_k^{\eta - 1} = \gamma_k U/g_k$, $\partial^2 G/\partial g_k\partial g_l =
(1-\eta)G\gamma_k\gamma_l/(g_k g_l) > 0$ for $k \ne l$ and $\eta < 0$; hence, because
$R(Q)>0$ is fixed, $\partial^2 V_R/\partial g_k\partial g_l =R(Q)A\,\partial^2
G/\partial g_k\partial g_l > 0$ and $\partial^2 V_R/\partial g_k\partial a_i =
R(Q)(\partial G/\partial g_k)(\partial A/\partial a_i) > 0$. With
$D(V_R)=(V_R/\bar V)^{\psi}$, $V_R^{\psi} \propto G^{\psi}A^{\psi}$ conditional on
$q$, and on distinct gate coordinates
$k \ne l$,
$\partial^2 G^{\psi}/\partial g_k\partial g_l = \psi G^{\psi}\gamma_k\gamma_l/(g_k
g_l)[\psi - \eta] > 0$ since $\eta < 0 < \psi$ (same algebra as the CES power
computation), so $V_R^{\psi}$ is supermodular in $g$. Fix the floor vector $g^{f}$
of the proposition ($\underline g_k\le g^{f}_k \le g_k$), and let
$Y(S)=V_R(q,S,a)^{\psi}$ for
sets $S$ of gate assets, where assets in $S$ are at their observed levels, gate
assets outside $S$ are held at the floor $g^{f}$, and
  conduct and the menu are held fixed; define $w(S)=Y(S)-Y(\varnothing)$. For two disjoint
groups of gate assets $S_1$ and $S_2$, the profiles for $S_1$ and $S_2$ have join
equal to the profile for $S_1 \cup S_2$ and meet equal to the all-floor profile,
so supermodularity gives
$Y(S_1\cup S_2)+Y(\varnothing)\ge Y(S_1)+Y(S_2)$, hence
$w(S_1\cup S_2)\ge w(S_1)+w(S_2)$. Applying the same argument inductively gives the
normalized superadditivity statement for any finite partition of the gate.
\end{proof}

\begin{proof}[Proof of Proposition~\ref{prop:shapley}]
Fix $q$ and let $R_q=R(Q(q))>0$. Let $z(t)=(tg,ta)$ for $t\in(0,1]$. By homogeneity of the two CES aggregators,
$G(tg)=tG(g)$ and $A(ta)=tA(a)$, while the shares $\gamma_k$ and $s_i$ are unchanged
along the diagonal. For a gate asset $k$,
\[
  \frac{\partial V_R(q,z(t))}{\partial z_k}
  = R_q\frac{G(tg)\gamma_k}{t g_k}\,A(ta)
  = t\,\frac{\gamma_k V_R}{g_k}.
\]
Thus
\[
  \alpha_k^{AS}
  = \int_0^1 g_k \frac{\partial V_R(q,z(t))}{\partial z_k}\,dt
  = \gamma_k V_R\int_0^1 t\,dt
  = \frac{1}{2}\gamma_k V_R .
\]
The same calculation for a menu asset $i$ gives
\[
  \frac{\partial V_R(q,z(t))}{\partial z_i}
  = R_qG(tg)\frac{A(ta)s_i}{t a_i}
  = t\,\frac{s_i V_R}{a_i},
  \qquad
  \alpha_i^{AS}=\frac{1}{2}s_i V_R .
\]
The allocation exhausts value because
$\sum_{k\in K}\alpha_k^{AS}+\sum_{i\in M}\alpha_i^{AS}
=\frac12 V_R\sum_k\gamma_k+\frac12 V_R\sum_i s_i=V_R$.
The final sentence of the proposition follows from Proposition~\ref{prop:asymmetry}(i):
the weakest gate asset's share satisfies $\gamma_k \ge \phi_k > 0$ independently of
the menu, while under menu regularity $s_i = O(1/|M|)$.

For completeness, if $q$ is scaled along the same ray, then
$Q(tq)=tQ(q)$ and the physical gate and menu allocations are respectively
$\gamma_kGA\int_0^1tR(tQ)dt$ and
$s_iGA\int_0^1tR(tQ)dt$. Conduct coordinate $q_k$ receives
$\chi_kGAQ\int_0^1t^2R'(tQ)dt$. Summing over all coordinates and using
\[
 \int_0^1\!\left[2tR(tQ)+t^2QR'(tQ)\right]dt
 =\left[t^2R(tQ)\right]_{0}^{1}=R(Q)
\]
shows that the expanded allocation exhausts $R(Q)GA=V_R$.
\end{proof}

\begin{proof}[Proof of Proposition~\ref{prop:commons}]
The custodian's interior first-order condition from \eqref{eq:custodian} is
\[
  \mu_k\frac{\partial N_k}{\partial g_k}
  = B_k'(1-g_k) + c_k'(g_k),
\]
where $B_k'$ is the marginal development benefit from lowering quality. Holding the
same degradation benefit and maintenance cost terms fixed, the domestic welfare
first-order condition replaces the private captured-arrival term with the marginal
domestic surplus from all affected regions and the non-use term:
\[
  \sum_j S'_j(N_j)\frac{\partial N_j}{\partial g_k}
  + \frac{\partial \zeta_k}{\partial g_k}
  = B_k'(1-g_k) + c_k'(g_k).
\]
Subtracting the private condition gives the welfare wedge
\[
  \tau_k^{W}
  = \sum_j S'_j(N_j)\frac{\partial N_j}{\partial g_k}
  - \mu_k\frac{\partial N_k}{\partial g_k}
  + \frac{\partial \zeta_k}{\partial g_k}.
\]
By the intensive-share condition stated in the proposition, write $N_j=h_jT_R$ with
$h_j$ independent of $g_k$, so that, conditional on $q$,
\[
  \frac{\partial N_j}{\partial g_k}
  = h_j\frac{\partial T_R}{\partial g_k}
  = h_j\,\psi\gamma_k\frac{T_R}{g_k},
\]
which is strictly positive for every region $j$ with positive intensive share. If
$S'_k(N_k)\ge \mu_k$ and $\partial\zeta_k/\partial g_k\ge0$ over the feasible interval,
every term in
\[
  \tau_k^W=(S'_k-\mu_k)\frac{\partial N_k}{\partial g_k}
  +\sum_{j\ne k}S'_j\frac{\partial N_j}{\partial g_k}
  +\frac{\partial\zeta_k}{\partial g_k}
\]
is weakly nonnegative. The wedge is strictly positive under any of the strictness
conditions stated in the proposition.

For the level comparison, hold $(q,g_{-k},a)$ fixed and let $F(g_k)$ denote the
custodian's objective \eqref{eq:custodian} and $W(g_k)$ the domestic welfare
objective, so that $W'(g_k) - F'(g_k) = \tau_k^{W}(g_k)$ by the two first-order
expressions above. Because $\tau_k^W\ge0$, $W-F$ is nondecreasing, so the binary
objective family consisting of $F$ and $W$ has increasing differences in quality
and the private/social index. Standard monotone comparative statics on the compact
quality interval imply that $\arg\max W$ dominates $\arg\max F$ in the strong set
order. Hence their smallest and largest selections are weakly ordered; with unique
maximizers, $g_k^*\ge g_k^{dec}$. If $\tau_k^W>0$ throughout the interval between
unique interior optima, equality is impossible because $F'(g_k^{dec})=0$ whereas
$F'(g_k^*)=-\tau_k^W(g_k^*)<0$ whenever $W'(g_k^*)=0$.

For implementation, a transfer schedule with derivative $\tau_k^W$ at the target
adds the required marginal benefit and makes the custodian's first-order condition
coincide with the domestic welfare condition there. A per-arrival revenue-sharing
increment $\Delta\mu_k$ instead contributes
$\Delta\mu_k\,\partial N_k/\partial g_k$, so it must satisfy the conversion in the
proposition. If each subsidized objective is single-peaked, the target is a best
response for every custodian and hence an equilibrium. The claim is first-order
implementation at the target, not global uniqueness.

For finite $\eta < 0$ every gate share is strictly positive and $h_k>0$, so
$\partial N_k/\partial g_k = h_k\,\psi\gamma_k T_R/g_k > 0$ for every custodian.
Among otherwise identical custodians -- same functional forms -- the objective
$\pi_k$ then has increasing differences in $(g_k, \mu_k)$, since
$\partial^2 \pi_k/\partial g_k\,\partial \mu_k = \partial N_k/\partial g_k > 0$,
and decreasing differences in $g_k$ and any upward shift of $B_k'$ or $c_k'$; by
monotone comparative statics the quality argmax correspondence is nondecreasing in
$\mu_k$ and nonincreasing under such shifts. On the compact normalized quality
domain, any convergent selections preserve these weak orders as $\eta\to-\infty$.
The limiting tourism index depends on $\min_k g_k$
(Proposition~\ref{prop:weakerlink}) and is bound by the lowest selected quality. This
is a ceteris-paribus statement, not a global ranking of arbitrary equilibria.
\end{proof}

\end{document}